\documentclass{llncs}

\usepackage{microtype}

\usepackage{latexsym,amssymb,amsmath}

\usepackage{cite}

\usepackage{wrapfig} % Package for putting figure in text paragraphs.

\usepackage{graphicx}%[dvips]{graphicx}

\newtheorem{observation}{Observation}

\usepackage[usenames,dvipsnames]{color}

\DeclareMathOperator{\RG}{RG} %% Can use any other notation you want!
\DeclareMathOperator{\RIG}{RIG} %% Can use any other notation you want!
\DeclareMathOperator{\GG}{GG} %% Can use any other notation you want!
\DeclareMathOperator{\DG}{DG} %% Can use any other notation you want!
\usepackage{paralist}
\title{Witness Rectangle Graphs%
  \thanks{Research of B.A.\ has been partially supported by NSA MSP
    Grants H98230-06-1-0016 and H98230-10-1-0210.  Research of B.A.\
    and M.D. has also been supported by a grant from the U.S.-Israel
    Binational Science Foundation and by NSF Grant CCF-08-30691.
    Research by F.H. has been partially supported by projects MEC
    MTM2006-01267 and MTM2009-07242, and Gen.~Catalunya DGR
    2005SGR00692 and 2009SGR1040.
    A preliminary version of this paper appeared in the \emph{Proceedings of
  the Algorithms and Data Structures Symposium (WADS'11)} \cite{ADH11}.}}

\author{Boris Aronov\inst{1}%\thanks{%
  \and
  Muriel Dulieu\inst{1}
  \and
  Ferran Hurtado\inst{2}
    }

\institute{Department of Computer Science and Engineering, Polytechnic 
 Institute of NYU, Brooklyn, NY~11201-3840, USA;
 \email{aronov@poly.edu}, \email{mdulieu@gmail.com}
 \and Departament de Matem\`{a}tica Aplicada~II,
    Universitat Polit\`{e}cnica de Catalunya (UPC),
    Barcelona, Spain. \email{ferran.hurtado@upc.edu}}

\makeatletter
\renewcommand\section{\@startsection{section}{1}{\z@}%
                       {-8\p@ \@plus -2\p@ \@minus -2\p@}%
                       {6\p@ \@plus 2\p@ \@minus 2\p@}%
                       {\normalfont\large\bfseries\boldmath
                        \rightskip=\z@ \@plus 4em\pretolerance=10000 }}
\makeatother

\begin{document}
% -----------------------------------------------------------------------------
\maketitle

%\begin{titlepage}
%  \maketitle \thispagestyle{empty}
  \begin{abstract}
    In a \emph{witness rectangle graph} (WRG) on vertex point set $P$
    with respect to witness point set~$W$ in the plane, two points
    $x,y$ in $P$ are adjacent whenever the open isothetic rectangle with $x$ and
    $y$ as opposite corners contains at least one point in
    $W$.  WRGs are representative of a larger family of witness
    proximity graphs introduced in two previous papers.

    %% Just dropped. --BA
    % We describe how to compute $\RG^+(P,W)$ for given sets $P$ and $W$
    % in optimal output-sensitive way. \textbf{MODIFY THIS SENTENCE}

    We study graph-theoretic properties of WRGs.  We prove that
    any WRG has at most two non-trivial connected components.  We
    bound the diameter of the non-trivial connected components of a
    WRG in both the one-component and two-component cases.  In the
    latter case, we prove that a graph is representable as a
    WRG if and only if each component is a connected co-interval graph, thereby
    providing a complete characterization of WRGs of this type.  We
    also completely characterize trees drawable as WRGs. 
    In addition, we prove that a WRG with no isolated vertices has domination number at most four.
    
    Moreover, we show that any combinatorial graph can be drawn as a WRG using a combination of 
    positive and negative witnesses.

    Finally, we conclude with some related results on the number of
    points required to stab all the rectangles defined by a set of
    $n$~points.
  \end{abstract}
% \end{titlepage}
%%%%%%%%%%%%%%%%%%%%%%%%%%%%%%%%%%%%%%%%%%%%%%%%%%%%%%%%%%%%%%%%%%%%%%%%%%%%%%

\pagestyle{plain}

\section{Introduction}\label{section:introduction}

\emph{Proximity graphs} have been widely used in situations in which there is a need of expressing the fact that some objects in a given set---which are
assigned to nodes in the graph---are close, adjacent, or neighbors, according to some geometric, physical, or conceptual
criteria, which translates to edges being added to the corresponding graph.
In the geometric scenario the objects are often points and the goal is to analyze the shape or the structure of the set of spatial data they describe or represent. This situation arises in Computer Vision, Pattern Recognition, Geographic Information Systems, and Data Mining, among other fields. The paper \cite{JT92} is a survey from this viewpoint, and several related papers appear in \cite{T88}. In most proximity graphs, given a point set $P$, the adjacency between two points $p,q\in P$ is decided by checking
whether in their \emph{region of influence} there is no
other point from $P$, besides $p$ and $q$. One may say that the presence of another point
is considered an \emph{interference}. There are many variations,
depending %mainly 
on the choice of the family of influence regions \cite{JT92,Co06,Li08}.

%Maybe the most popular proximity graph is the $k$-\emph{Nearest Neighbor Graph}, a basic tool in classification methods, in which every point is connected to its $k$ closest neighbors.

Given a combinatorial graph $G=(V,E)$, a \emph{proximity drawing} of $G$ consists of
a choice of a point set $P$ in the plane with $|P|=|V|$, for a given criterion of neighborhood
for points, such that the corresponding proximity graph on $P$ is
isomorphic to $G$.  This question
belongs to the subject of \emph{graph drawing problems}, in which the emphasis is on geometrically representing graphs
with good readability properties and fulfilling some aesthetic criteria \cite{BETT98}. The main issues %in this respect 
are to characterize the graphs that admit a certain kind of representation, and to design efficient algorithms for finding such a drawing,
whenever possible.

Proximity drawings have been studied extensively and utilized widely \cite{BLL94,Li08}. However, this  kind of representation is
somehow limited and there have been %already 
some attempts to expand the class, for example using weak proximity graphs \cite{BLW06}.
% Another generalization recently introduced, which is the most relevant for this paper, is the
% concept of \emph{witness proximity graph}, in which the adjacency of
% points in a
Another recently introduced generalization is the
concept of \emph{witness proximity graphs} \cite{ADH,ADH08}, in which the adjacency of points in a
given vertex set $P$ is decided by the presence or
absence of points from a second
point set $W$---the \emph{witnesses}---in their region of
influence. This generalization includes the classic proximity graphs
as a particular case,
and offers both a stronger tool for neighborhood description and much more flexibility for graph representation purposes.

In the \emph{positive witness} version, there is an adjacency when
a witness point is covered by the region of influence.  In the \emph{negative witness} version,
two points are neighbors when they admit a region of influence free of any witnesses.
In both cases the decision is based on the presence or absence of witnesses
in the regions of influence, and a combination of both types of
witnesses may also be considered; refer to Section~\ref{sec:PlusMinusW}. Observe that by taking $W=P$ playing a negative role, we recover the original proximity
graphs; so this is a proper generalization.
Witness graphs were introduced in \cite{ADH}, where the focus is on the generalization of \emph{Delaunay graphs}. The witness version
of \emph{Gabriel graphs} was studied in \cite{ADH08}, and a thorough exploration of this set of problems is the main topic
of the thesis \cite{thesis}.

In this paper, we consider a positive witness proximity
graph related to the rectangle-of-influence graph, the
witness rectangle graph.  The \emph{rectangle of influence graph}
$\RIG(P)$, also named the Delaunay graph of a point set in the plane with respect to axis-parallel rectangles~\cite{Janos}, is usually studied as one of the basic proximity graphs
\cite{LLMW98,Li08}.
In $\RIG(P)$, $x,y \in P$ are adjacent when the \emph{rectangle
  $B(x,y)$ they define} covers no third point from $P$; $B(x,y)$ is
the unique open isothetic rectangle with $x$ and $y$ at its opposite corners.
In other words, two points are adjacent to each other when they are rectangularly visible~\cite{RectangularVisibility1, RectangularVisibility2, RectangularVisibility3}.
Much effort has been devoted to the rectangle of influence drawability problem~\cite{LLMW98, RectangleDrawability1, RectangleDrawability2, RectangleDrawability3, RectangleDrawability4, RectangleDrawability5, RectangleDrawability6}.
%
%% BA: defined above
% Given two such points $x$ and $y$, the
% \emph{rectangle $B(x,y)$ defined by $x$ and $y$} is the open isothetic
% rectangle with $x$ and $y$ at its opposite corners.
The \emph{witness rectangle graph}~(WRG) of vertex point set $P$ (or,
simply, \emph{vertices}) with respect to witness point set $W$
(\emph{witnesses}), denoted $\RG^+(P,W)$, is the graph with the vertex
set $P$, in which two points $x,y\in P$ are adjacent when the
rectangle $B(x,y)$ contains at least one point of $W$.
% \muriel{This would be the
%  place the mention that witness rectangle graphs are not new, Ichino
%  and Slansky in \cite{MNG} defined them already under the name
%  "mutual neighborhood graph"}
% Blocked-rectangle graphs
% are representatives of a larger family of graphs, namely witness
% graphs, that \boris{do we discuss it above, or do we just mention it
%   in passing here?}  Specifically, a witness rectangle graph is a
% positive-witness graph in which the shapes are the rectangles defined
% by pairs of points from $P$.
The graph $\RG^+(P,\varnothing)$ has no edges.  When $W$ is
sufficiently large and appropriately distributed, $\RG^+(P,W)$ is
complete. 
%\footnote{We assume throughout most of our discussion that no two
%  points of $P \cup W$ lie on the same vertical or the same horizontal
%  line.}
% \muriel{This would be the place to mention that Ichino and Slansky had invented a negative version of 
% this graph before us.}
%\Boris{Changed.  Next two sentence do not read well.}
Ichino and Slansky's \cite{MNG} \emph{mutual neighborhood graph} 
is precisely the negative-witness version of the witness rectangle 
graph. See further discussion in Section~\ref{sec:MNG}.
We also note that a
negative-witness version of this graph with $W=P$ would be precisely
$\RIG(P)$ discussed above;
in fact $\RG^+(P,P)$ is precisely the complement of
$\RIG(P)$.     An example is shown in
Figure~\ref{fig:RGexample}. 

We show in this paper that the connected components of WRGs are geometric examples of graphs with small diameter; these have
  been attracting attention in the pure graph theory setting \cite{SmallDiameter1, SmallDiameter2}, and are far from being
  well understood, even for diameter two \cite{diameterTwo1,diameterTwo2,diameterTwo3}.
  We prove below that the maximum domination number of a WRG with no isolated vertices is four, which we find interesting, considering that the domination number of sufficiently large \emph{planar} graphs of diameter three is at most seven~\cite{Domination, Domination2}.

\begin{figure}[htbp!]
  \centering
  \includegraphics[scale=0.8]{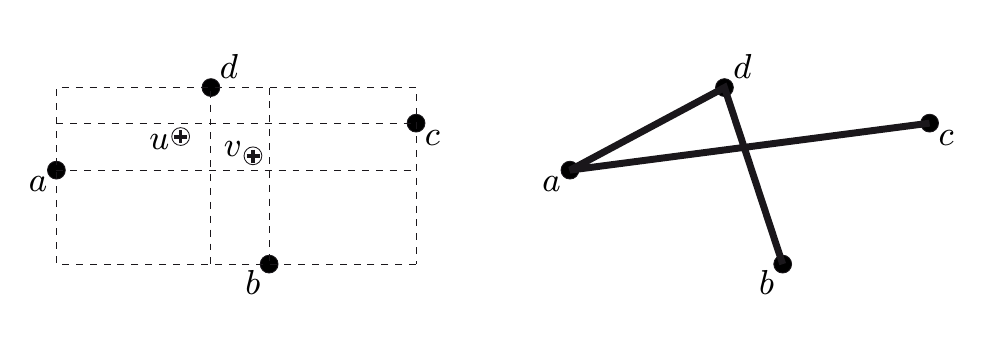}
  \caption{Left: A set of points $P=\{a,b,c,d\}$ and a witness set
    $W=\{u,v\}$. Right: the witness rectangle graph $\RG^+(P,W)$.  In
    all our figures for WRGs solid dots denote vertices and dots with
    a cross denote (positive) witnesses.  }
  \label{fig:RGexample}
\end{figure}

Besides some computational issues, such as the construction of $\RG^+(P,W)$ for given sets $P$
and $W$ in an output-sensitive manner, %in this paper 
we study several graph-theoretic properties of WRGs: (a)~We
completely characterize trees drawable as WRGs (Theorem~\ref{tree}).
(b)~We argue that any WRG has zero, one, or two \emph{non-trivial} connected
components (see the definition below and Theorem~\ref{thm:rectangleStructural}). (c)~We prove
that the diameter of a single-component WRG is at most six, and that this bound is tight in the worst case
(Theorem~\ref{thm:rectangleStructural} and subsequent discussion). (d)~We prove that the diameter
of a (non-trivial) connected component of a two-component WRG is at
most three and this can be achieved in the worst case
(Theorem~\ref{thm:rectangleStructural}).
(e)~In the two-component case, we provide a complete characterization
of graphs representable as WRGs. Such graphs, disregarding
isolated vertices, are precisely disjoint unions of two connected co-interval
graphs (Theorem~\ref{staircase}).
This last result allows us to recognize in linear time if a combinatorial graph with two
non-trivial components can be drawn as a WRG, and to construct such a drawing if it exists.
(f)~We prove that the maximum domination number of a WRG with no isolated vertices is four.

%\Boris{update}
%Are staircase graphs easy to deal with, in graph-theoretic sense?  We
%partially answer the question, as follows: We
%% give a weak
%% graph-theoretic characterization of the staircase graphs (\Boris{see
%%   ...}) and
%present in Section~\ref{section:staircase} an algorithm which, given a (combinatorial)
%staircase graph, produces a realization of it as staircase drawing---a geometric
%staircase graph of specific point sets as vertices and witnesses.  For future work, we hope to refine our
%algorithm so that it can be given any combinatorial graph~$G$ as
%input and would construct a staircase drawing realization of~$G$ if
%one exists and would produce evidence that~$G$ is not a staircase
%graph otherwise.

In Section~\ref{sec:MNG}, we give a counterexample to a theorem of
Ichino and Slansky in~\cite{MNG}.
%  about mutual neighborhood graphs, which are precisely witness rectangle graphs.
% \muriel{Boris wants to do something about this}
In Section~\ref{sec:PlusMinusW}, we show that any combinatorial graph can be drawn as a WRG with positive and negative witnesses, using a quadratic number of witnesses.
Finally, in Section~\ref{section:stabbing}, we present some related
results on stabbing rectangles defined by a set of points with other
points.  They can be interpreted as questions on ``blocking''
rectangular influences.

\paragraph*{Terminology and notation.}
Throughout the paper, we will work with finite point sets in the
plane, in which no two points lie on the same vertical or the same
horizontal line.

Hereafter, for a graph $G=(V,E)$ we write $xy\in E$ or $x\sim y$ to
indicate that $x,y\in V$ are adjacent in $G$, and generally use
standard graph terminology as in \cite{CL04}.
% In particular, we use
% $P_t$ to denote the path graph on $t$~vertices.
%
When we speak of a \emph{non-trivial connected component} of a graph,
we refer to a connected component with at least one edge (and at least
two vertices).

Given two graphs $G_1=(V_1,E_1)$ and $G_2=(V_2,E_2)$ with disjoint
vertex sets, their \emph{join} is the graph $G_1+G_2=(V_1 \cup V_2,
E_1\cup E_2 \cup V_1 \times V_2)$~\cite{mathworld:join}.

%%%%%%%%%%%%%%%%%%%%%%%%%%%%%%%%%%%%%%%%%%%%%%
\paragraph*{How to compute a witness rectangle graph.}

%\boris{I find it odd that we define dominance here, use it for one
%  theorem, and never use it anywhere else, even though half of the
%  statements in the paper can in phrased in terms of dominance.  Just a
%  comment.} 
De Berg, Carlsson, and Overmars \cite{BCO92} generalized the notion of 
dominance, in a way that is closely related to witness rectangle graphs, by defining dominance pairs $p,q$ of a set of points $P$, with respect to 
a set $O$ of so-called obstacle points.
More precisely, $p$ is said to \emph{dominate $q$ with respect to $O$} if there is no point $o \in O$ such that $p$ dominates $o$ and $o$ dominates $q$.
Recall that $p$ \emph{dominates} $q$ if and only if $x(p) \geq x(q)$,
$y(p) \geq y(q)$, and $p \neq q$.

They prove the following theorem:

%De Berg, Carlsson, and Overmars \cite{BCO92} considered the notion of
%``rectangular visibility with obstacles,'' which is closely related to
%the subject of this paper.  They compute, given two point sets $P$ and
%$W$, the set of all pairs $x,y \in P$, so that $B(x,y)$ contains no
%point of $W$, in time $O(k+n\log n)$, where $k$ is the number of such
%pairs and $n:=\max\{|P|,|W|\}$.  In our terminology, this is precisely
%the negative-witness graph of $P$ with witnesses $W$, with respect to
%the family of boxes defined by $P$.  We are interested in its positive
%version.  A simple modification of their algorithm yields:

\begin{theorem}[De Berg, Carlsson, and Overmars \cite{BCO92}]
All dominance pairs in a set of points $P$ with respect to a set of points $O$ can be computed in time 
$O(n \log n + k)$, where $n = |P| + |O|$ and $k$ is the number of answers.
\end{theorem}

Collecting all dominance pairs in a set of points $P$ with respect to a set of points $W$, and repeating the procedure after rotating the plane by $90^\circ$, one obtains the negative version of the witness rectangle graph.
A simple modification of their algorithm yields the positive version:

\begin{corollary}
  Let $P$ and $W$ be two point sets in the plane.  The
  witness rectangle graph $\RG^+(P,W)$ with $k$~edges can be computed
  in $O(n \log n + k)$ time, where $n:=\max\{|P|,|W|\}$.
\end{corollary}

%%%%%%%%%%%%%%%%%%%%%%%%%%%%%%%%%%%%%%%%%%%%%%%%%%%%%%%%%%%%%%%%%%%%%%%%%%%%%%
\section{Structure of Witness Rectangle Graphs}\label{section:rectangles}

Let $G:=\RG^+(P,W)$ be the witness rectangle graph of vertex set $P$
with respect to witness set $W$.
We assume that the set of witnesses is \emph{minimal}, in the sense that
removing any witness from $W$ changes $G$.  Put
$n:=\max\{|P|,|W|\}$ and let $E:=E(G)$ be the edge set of $G$.  We
partition $E$ into $E^+$ and $E^-$ according to the slope sign of the
edges when drawn as segments.  Slightly abusing the terminology we refer
to two edges of $E^+$ (or two edges of $E^-$) as having \emph{the same
  slope} and an edge of $E^+$ and an edge of $E^-$ as having
\emph{opposite slopes}.

\begin{wrapfigure}[7]{r}{0pt}
  % \vspace{-1ex}%
  % \hspace{0pt plus 1 fill}
  \includegraphics[width=4.5cm]{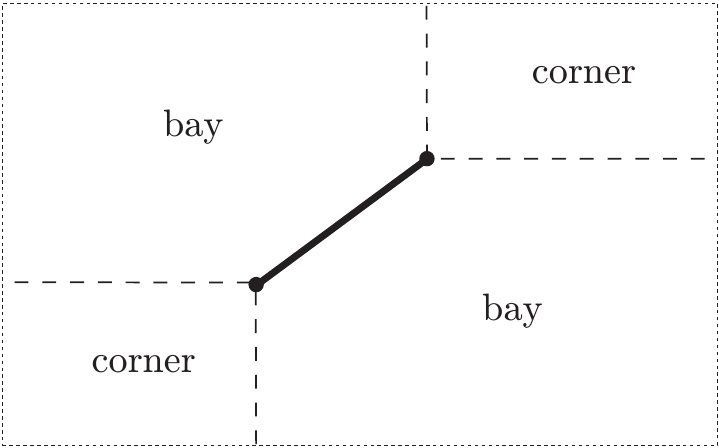}
\end{wrapfigure}
Recall that the open isothetic rectangle (or \emph{box}, for short) defined
by two points $p$ and $q$ in the plane is denoted $B(p,q)$; for an
edge $e=pq$ we also write $B(e)$ instead of $B(p,q)$.  Every edge $e$,
say in $E^+$, defines four regions as in the figure on the right,
%\ref{fig:cornerBay},
that we call \emph{(open) corners} and \emph{(closed) bays}.
% \muriel{should we specify that these are open and closed sets respectively?}
% \begin{figure}[htbp!]
%   \centering
%   \includegraphics[scale=0.75]{cornerBay}
%   \caption{Corners and bays.}
%   \label{fig:cornerBay}
% \end{figure}

% An illustration is shown in
% Figure~\ref{fig:treesAndpath}.

% \begin{figure}[htbp!]
%   \centering
%   \includegraphics[scale=0.75]{treesAndPAth}
%   \caption{Tree of diameter 3 and $P_7$.}
%   \label{fig:treesAndpath}
% \end{figure}

\begin{observation} \label{obs:quadrants} Every $x\in P$ inside a
  corner of an edge $e$ is adjacent to at least one endpoint of $e$.
\end{observation}

% \begin{wrapfigure}{r}{0pt}
%   \includegraphics[width=7cm]{treesAndPAth}
%   \caption{Tree of diameter 3 and $P_7$.}
%   \label{fig:treesAndpath}
% \end{wrapfigure}
% Before delving into analysis of properties of RGs, we point
% out that all trees with diameter at most three (i.e., stars and double
% stars) and all paths on at most seven nodes are RGs; refer to figure
% on the right.
Note that for any $P$, $W$, and $P' \subset P$, the graph $\RG^+(P',W)$
is an induced subgraph of $\RG^+(P,W)$, so the class of graphs
representable as WRGs is closed under the operation of taking induced
subgraphs.

%% BA: TO DO? This says they form an induced 2K_2.
%% The following statements refer to absence of an induced 3K_2.
Two edges are \emph{independent} when they share no vertices and the
subgraph induced by their endpoints contains no third edge.  Below we
show that $G$ cannot contain three pairwise independent edges, which
imposes severe constraints on the graph structure of $G$.

\begin{lemma}
  \label{lem:2sameSlope}
  Two independent edges in $E^+$ (respectively, $E^-$) cannot cross or
  share a witness.  The line defined by their witnesses is of negative
  (respectively, positive) slope.
\end {lemma}
\begin{proof} Let the two edges be $ab,cd\in E^+$, with $x(a)<x(b)$
  and $x(c)<x(d)$.  A common witness would have $a$ and $c$ in its
  third quadrant and $b$ and $d$ in the first, implying $a\sim d$ and
  $c\sim b$, a contradiction.  If $ab$ and $cd$ cross, assume
  without loss of generality that $x(a)<x(c)$.  Neither $c$ nor $d$
  can be inside $B(a,b)$ (because of Observation~\ref{obs:quadrants})
  and hence $B(a,d)\cup B(c,b)\supset B(a,b)$, implying $a\sim d$ or
  $c\sim b$, a contradiction.  Finally, the second part of the
  statement is a direct consequence of
  Observation~\ref{obs:quadrants}. \hfill $\Box$
\end{proof}

\begin{lemma}
  \label{lem:2oppositeSlope}
  Two independent edges with opposite slopes must share a witness.
\end {lemma}
\begin{proof}
  Let $ab\in E^+$ and $cd \in E^-$ be independent.  Let $w$ be a
  witness for $ab$ and let $w'$ be a witness for $cd$.  The points $c$
  and $d$ are not in quadrants I~or~III of $w$, as otherwise the two
  edges would not be independent.
%% WAS:
% It cannot be that $c$ lies in
% quadrant~II of $w$ while $d$ lies in quadrant~IV, or vice versa, as
% $w$ would be shared.
  If $w$ is shared, we are done.  Otherwise it cannot be that $c$ lies
  in quadrant~II of $w$ while $d$ lies in quadrant~IV, or vice versa.
  Therefore $c$ and $d$ are either both in quadrant~II or both in
  quadrant~IV of $w$.  Assume, without loss of generality, the former
  is true.  The witness $w'$ is not outside of $B(a,b)$, as we would
  have $c \sim a$ and/or $c \sim b$ (assuming, without loss of
  generality, that $x(c)<x(d)$) and the edges would not be
  independent.  Therefore $w'$ is in $B(a,b)$, so $w'$ is a shared
  witness, as claimed. 
    \hspace*{0pt} \hfill $\Box$
%
  % \ferran{Ferran's proof} If $y(w)<y(w')$ then $a$ is in the region
  % $R_7$ of Figure \ref{fig:2oppositeSlope}, left, and hence $c$ is
  % in region $R_1$ and $d$ is in region $R_9$; therefore $w$ is also
  % a witness for $cd$.  If $y(w)>y(w')$ then $d$ is in the region
  % $R_9$ of Figure \ref{fig:2oppositeSlope}, right, and therefore $a$
  % is in region $R_7$ and $b$ is in region $R_3$; therefore $w'$ is
  % also a witness for $ab$.  \boris{Can this be done easier?}
\end{proof}

%\begin{figure}[htbp!]
%  \centering
%  \includegraphics[scale=0.75]{2oppositeSlope}
%  \caption{Two cases in the proof of Lemma~\ref{lem:2oppositeSlope}.}
%  \label{fig:2oppositeSlope}
%  \boris{The apostrophies in $w$' should be primes: $w'$! I cannot
%    open or edit Illustrator files.}
%\end{figure}

\begin{lemma}
  \label{lem:no3independentSameSlope}
  There are no three pairwise independent edges in $E^+$ (or in
  $E^-$).
\end{lemma}
\begin{proof} Assume that three pairwise independent edges $e_1, e_2,
  e_3$ in $E^+$ are witnessed by $w_1,w_2,w_3$, respectively, with
  $x(w_1)<x(w_2)<x(w_3)$.  Then, by Lemma~\ref{lem:2sameSlope},
  $y(w_1)>y(w_2)>y(w_3)$.  By the same lemma at least one endpoint of
  $e_1$ is in the second quadrant of $w_2$ and at least one endpoint
  of $e_3$ is in its fourth quadrant, contradicting their independence. \hfill $\Box$
%  \muriel{Is this obvious that this is true due to Lemma 1?}
\end{proof} 

\begin{lemma}
  \label{lem:no3independent}
  A witness rectangle graph does not contain three pairwise independent edges.
\end{lemma}

\begin{proof}
  By Lemma~\ref{lem:no3independentSameSlope}, two edges $ab$ and $cd$
  of the three pairwise independent edges $ab$, $cd$, and $ef$ have
  opposite slopes.  By Lemma~\ref{lem:2oppositeSlope}, $ab$ and $cd$
  share a witness point $w$.  Every quadrant of $w$ contains one of
  the points $a$, $b$, $c$, or $d$, therefore both $e$ and $f$ must be
  adjacent to one of them, a contradiction.  \hfill $\Box$
%
  % The proof is by contradiction.  Let $e_1, e_2, e_3$ be three
  % pairwise independent edges; by Lemma
  % \ref{lem:no3independentSameSlope} they cannot have all the same
  % slope sign.  Assume that $e_1=ab$ and $e_2=cd$ have negative and
  % positive slope, respectively, with $x(a)<x(b)$ and $x(c)<x(d)$,
  % and that $w$ is a common witness for $e_1$ and $e_2$, whose
  % existence has been proved in Lemma~\ref{lem:2oppositeSlope}.  The
  % witness $w$ must lie inside the box $B(a,b)$, as shown in Figure
  % \ref{fig:no3independent}.
%
  % Now $a$, $b$, $c$ and $d$ lie in the second, fourth, third and
  % first quadrant of $w$, respectively, and therefore any other
  % vertex of the graph is adjacent to at least one of them, which
  % leaves no room for the endpoints of $e_3$, contradicting the
  % assumed independence.
%
\end{proof}

%\begin{figure}[htbp!]
%  \centering
%  \includegraphics[scale=0.75]{no3independent}
%  \caption{Illustration for the proof of Lemma~\ref{lem:2oppositeSlope}.}
%  \label{fig:no3independent}
%\end{figure}

The preceding results allow a complete characterization of the trees that can
be realized as WRGs. An analogous result for rectangle-of-influence graphs was
given in \cite{LLMW98}.

\begin{theorem}
  \label{tree}
  A tree is representable as a witness rectangle graph if and only if it has no three
  independent edges.
  % \begin{figure}[htbp!]
  %   \centering
  %   \includegraphics[width=4cm]{RGtree5}
  %   \includegraphics[width=4cm]{RGtree6}
  %   \caption{Any number of vertices can be added in the regions
  %     containing three vertices.}
  %   \label{RGtree}
  % \end{figure}
\end{theorem}
%{\def\proofname{Sketch of the proof}
\begin{proof}
  %Examine all combinatorial trees without three independent edges (see Figure~\ref{RGtree}). 
  %\begin{wrapfigure}[12]{R}{0pt}
\begin{figure}
\centering
  \includegraphics[width=1.8cm]{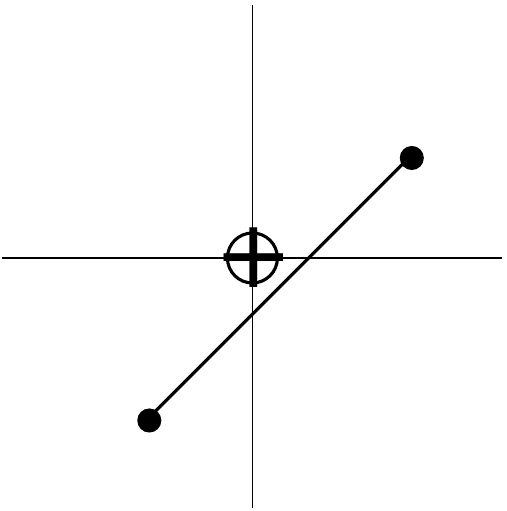}\hfil
  \includegraphics[width=1.8cm]{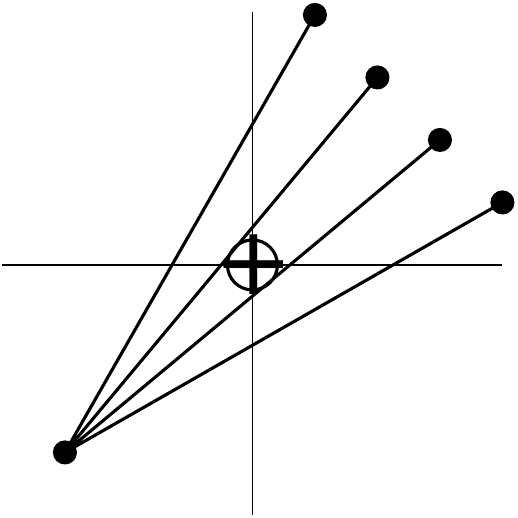}\hfil
  \includegraphics[width=2.7cm]{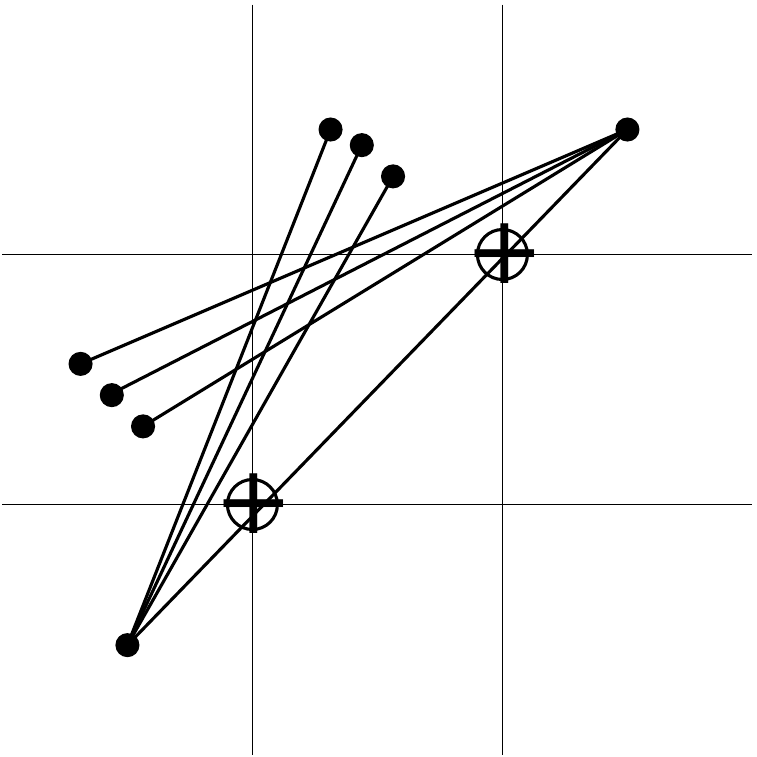}\hfil
  \includegraphics[width=3.6cm]{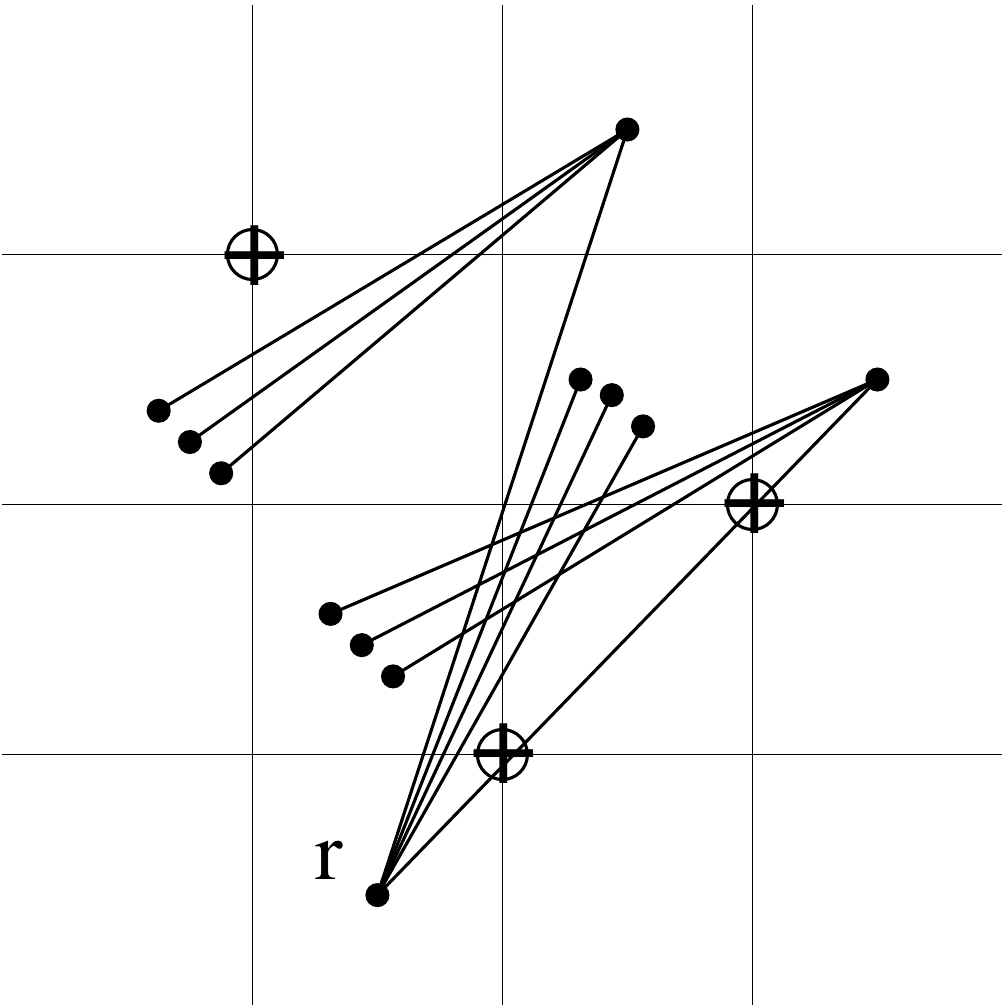}\\[2ex]
  \includegraphics[width=4.5cm]{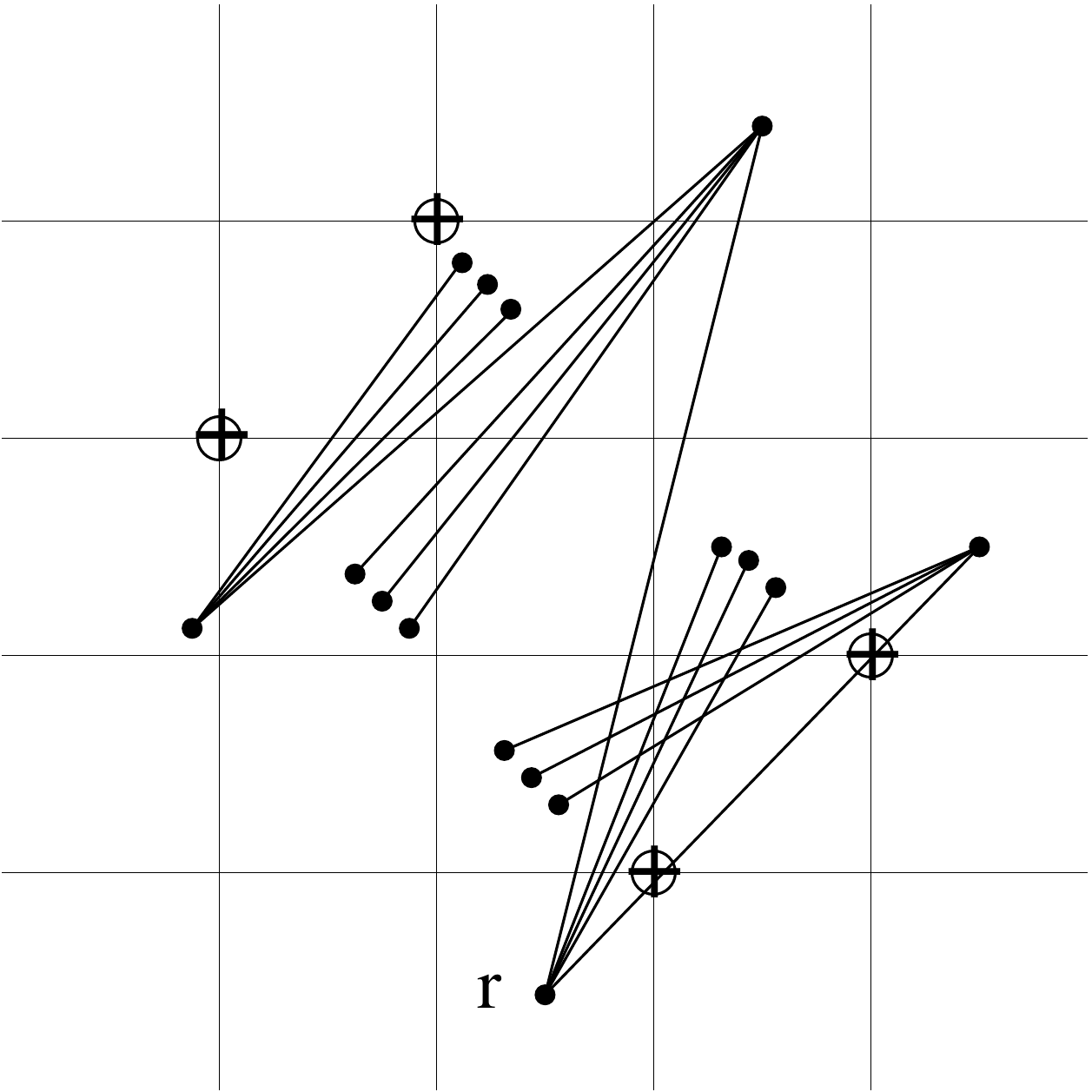}\hfil%
  \includegraphics[width=4.5cm]{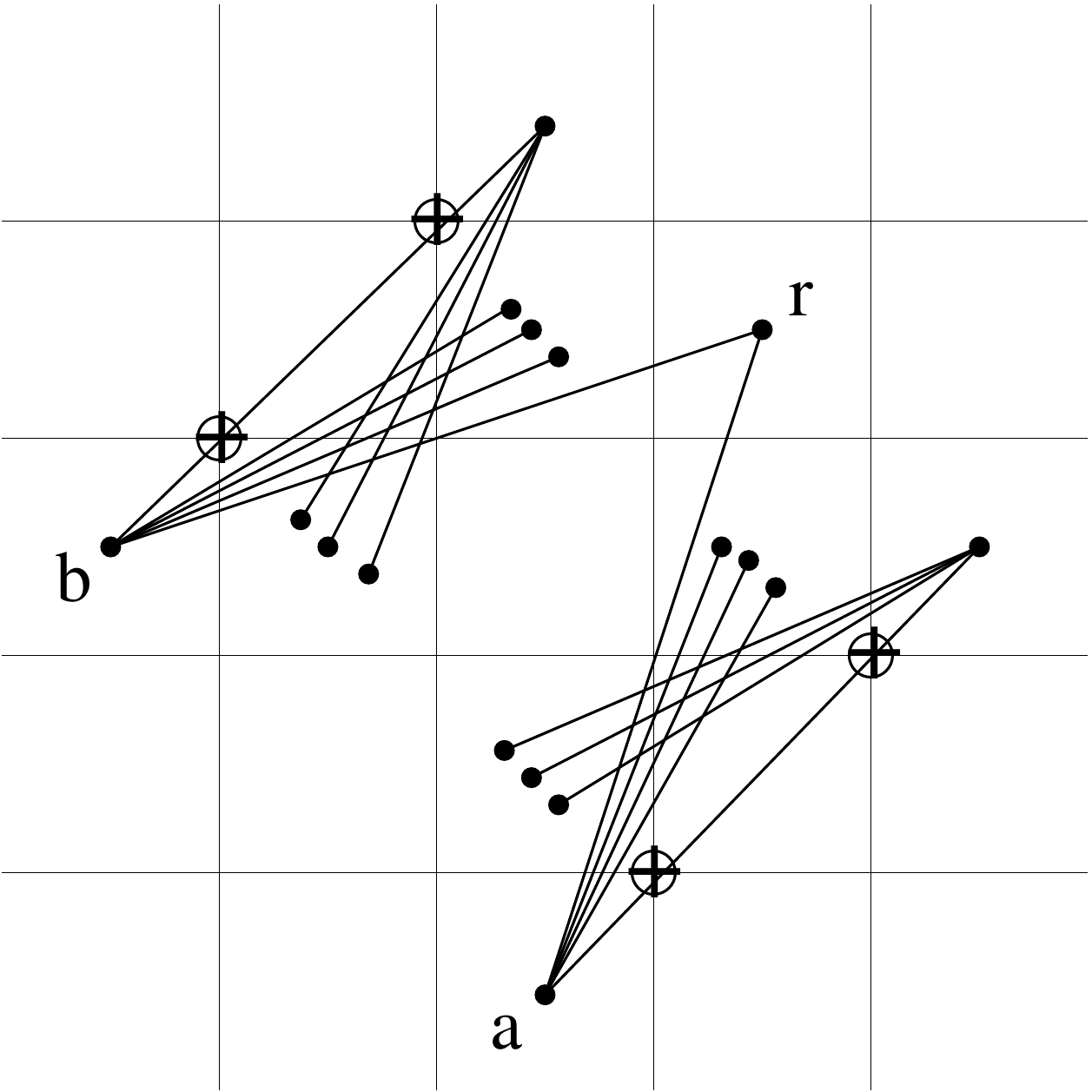}
  \caption{From left to right, All WRG trees of diameter 1 to 6.  Any number of vertices can
    be added in the cells containing three vertices.}
      \label{RGtree}
  \end{figure}
%\end{wrapfigure}
% \Boris{Should we center the little figures in the top row
%   vertically, or put the first two on top of each other, or
%   something?}  
In Figure~\ref{RGtree}, all maximal WRG trees are
represented. They are ordered by their diameter from left to
right. Any number of vertices can be added in the cells where there
are at least three vertices.

As one can see in Figure~\ref{RGtree} (upper left), all trees of diameter one to three are realizable as WRG.

The unique maximal WRG tree of diameter four (Figure~\ref{RGtree} upper right) is obtained by checking all ways to add branches to a path of length four without increasing the diameter and without creating three independent edges. As one can see, branches of length one have been added to the interior vertices, adding a branch of length two would create three independent edges. Adding a branch incident to one of the end-vertices of the path would increase the diameter.

The unique maximal WRG tree of diameter five is constructed similarly (Figure~\ref{RGtree} bottom left). 
If $r$ is the root of the WRG tree, $r$ can have an arbitrary number of children, but at most two of its children can have children (otherwise there would be three independent edges), and at most one of its grandchildren can have children (otherwise the tree would have diameter six).
%  This requires a case
%  analysis that we omit in this version.
%  Any such tree is a subtree of one of the two maximal
%  trees depicted in Figure~\ref{RGtree}, both of which
%  happen to be representable as WRGs, as seen in the figure. 

The unique maximal WRG tree of diameter six is constructed similarly
to the case of diameter four (Figure~\ref{RGtree} bottom right).  If
$r$ is the root of the WRG tree, $r$ can have at most two children $a$
and $b$ (otherwise there would be three independent edges), an
arbitrary number of grandchildren adjacent to $a$ and $b$, and $a$
and $b$ can have an arbitrary number of grandchildren all adjacent to
exactly one of their children (otherwise there would be three
independent edges). \mbox{}\hfill $\Box$
\end{proof}

Lemma~\ref{lem:no3independent} immediately implies the following
structural result that is far from a complete characterization, yet
narrows substantially the class of graphs representable as WRGs.

\begin{theorem} \label{thm:rectangleStructural} A witness rectangle graph has at most two
  non-trivial connected components.  If there are exactly two, each
  has diameter at most three.  If there is one, its diameter is at
  most six.
\end{theorem}

Note that the bounds on the diameter are tight: the tree in
Figure~\ref{RGtree}~(right) has diameter six and it is easy to draw the
disjoint union of two three-link paths as a WRG, by merging all the vertices in the same cell and removing $r$ from  Figure~\ref{RGtree}~(right), for example.

\begin{corollary} 
The condition of Theorem~\ref{thm:rectangleStructural} is necessary to characterize a 
WRG but not sufficient. In other words, not all connected graphs with no three independent edges are realizable as a witness rectangle graph.
\end{corollary}
\begin{proof}
The graph formed from the complete graph $K_c$ by attaching an additional degree-one vertex adjacent to each of the $c$ vertices (see Figure~\ref{supernova}), which we call a \emph{$c$-supernova}, does not have \emph{two} independent edges.
\begin{figure}
\centering
\includegraphics[scale=0.5]{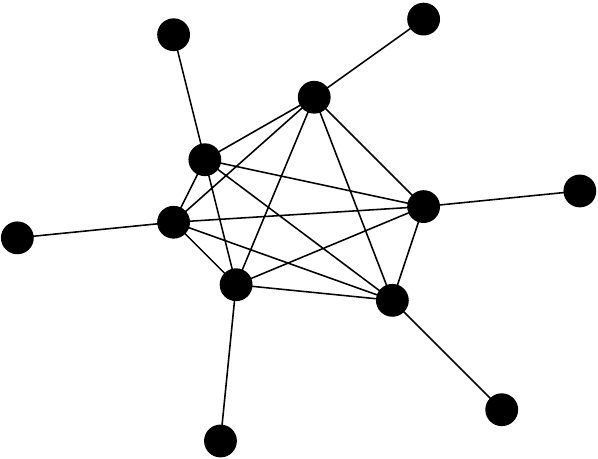}
\caption{A six-supernova.}
\label{supernova}
\end{figure}
Nevertheless, by Theorem~\ref{DominationNumberWRG}, a supernova with
$c > 4$ is not realizable as a WRG.  Indeed, to obtain a dominating
set of this graph, one must pick one vertex incident to each of the
$c$ edges incident to a leaf.  Therefore the domination number is at
least $c$ (actually it is $c$), which implies that graph is not
drawable as a WRG, by Theorem~\ref{DominationNumberWRG}.  \hfill
$\Box$
\end{proof}

In fact, a simple extension of the above argument yields a family of
graphs not realizable as WRGs.  Start with any graph with at least
five vertices and attach five edges to five of its vertices, as above.
The resulting graph (which may or may not have three independent
edges) has domination number at leave five and is therefore not
realizable.

\section{Two Connected Components}

In this section we define a subclass of witness rectangle graphs, called \emph{staircase graphs}.  We argue that a WRG with precisely two non-trivial
connected components
%\begin{wrapfigure}{r}{0pt}
 % \includegraphics[width=3cm]{BRGstaircase}
%\end{wrapfigure}
has a very rigid structure.  Namely, each of its non-trivial connected
components is isomorphic to a staircase graph.

\begin{definition}
  A \emph{staircase graph of type~IV} is a witness rectangle graph,
  such that the witnesses form an ascending chain (i.e., for every
  witness, other witnesses lie only in its quadrants I~and~III) and
  all the vertices lie above the chain (i.e., quadrant~IV of every
  witness is empty of vertices); refer to Figure~\ref{RGstaircase}.
\begin{figure}
\centering
  \includegraphics[width=4cm]{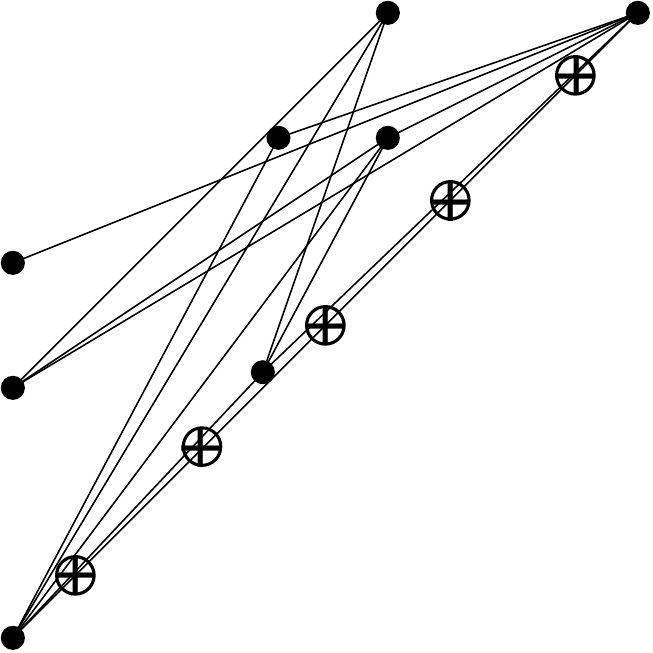}
   \caption{Staircase graph of type~IV.}
   \label{RGstaircase}
\end{figure}

  Staircase graphs of types~I,~II, and~III are defined analogously;
  they are rotated versions of the above.  The type of the staircase
  graph is determined by which quadrant of all witnesses is empty of
  vertices.
  %We call such a witness rectangle graph a \emph{staircase drawing}.
\end{definition}
% \begin{figure}
%   \centering
%   \includegraphics[width=3cm]{RGstaircase}
%   \caption{Staircase graph.}
%   \label{RGstaircase}
% \end{figure}
% \end{definition}

Note that an induced subgraph of a staircase graph is a staircase
graph (of the same type)---a property that immediately follows from
the definition and that we will find useful below.

\begin{lemma}
  \label{lem:staircase-splits}
  A combinatorial graph $G=(V,E)$, isomorphic to a staircase graph,
  is a join $G_1+G_2$ if and only if it can be realized
  as a staircase graph of type~IV with some witness containing points
  corresponding to $V(G_1)$ in quadrant~I and points corresponding to
  $V(G_2)$ in quadrant~III.
\end{lemma}

\begin{proof}
  %Refer to figure~\ref{staircaseGraph3}(left).
  Suppose $G=G_1+G_2$ is realizable as a staircase graph. 
  The combinatorial graphs $G_1$ and $G_2$ are
  isomorphic to staircase graphs $\RG^+(P_1,W_1)$ and $\RG^+(P_2,W_2)$
  of type~IV, which can be obtained, for example, by considering any
  realization of $G$ as a staircase graph of type~IV and dropping the
  points corresponding to $V(G_2)$ and $V(G_1)$, respectively.  Create
  a staircase graph of type~IV isomorphic to $G$ by placing a copy of
  $\RG^+(P_1,W_1)$ in quadrant~I of a new witness~$w$ and a copy of
  $\RG^+(P_2,W_2)$ in its quadrant~III.

  Conversely, given a staircase graph of type~IV isomorphic to $G$ such that some
  of its vertices (call the set $P_1$) are in the first quadrant of a
  witness $w \in W$, and the remaining vertices (call them $P_2$) are
  in its third quadrant, it is easily checked that $G=G_1+G_2$, where
  $G_1$ and $G_2$ are the subgraphs of $G$ induced by (the sets of
  vertices of $G$ corresponding to) $P_1$ and $P_2$, respectively.  \hfill $\Box$
\end{proof}

\begin{lemma}
  \label{quadEmpty}
  In a witness rectangle graph, if witness $w$ has no vertices in one of its quadrants,
  any witness in the empty quadrant is redundant.
\end{lemma}
\begin{proof}
  Suppose quadrant~II of $w$ contains a witness $w'$ but no vertices.
  Hence quadrant~II of $w'$ is empty of vertices as well.  Suppose
  $w'$ witnesses edge $ab$, and $a$ and $b$ are in its quadrants
  I~and~III respectively. (They cannot lie in quadrants II~and~IV, as
  quadrant~II is empty of vertices.)  As quadrant~IV of $w$ is
  included in quadrant~IV of $w'$, $a$ and $b$ must be in quadrants~I
  and III of $w$, respectively, as well.  Therefore $w$ witnesses
  $ab$.  Since this argument applies to all edges witnessed by $w'$,
  $w'$ is redundant.  \hfill $\Box$
\end{proof}

%\Boris[hints]{WARNING: Changed next lemma.  Check.}
\begin{lemma}
  \label{quadEmpty2}
%% BA: was:
  % In a witness rectangle graph $\RG^+(P,W)$, if, in a non-redundant%
  % \footnote{%
  %   Formally, $W'$ is non-redundant if and only if there is no proper
  %   subset $W''$ of $W'$ such that $\RG^+(P,W'')=\RG^+(P,W')$.}
  % %
  % set of witnesses $W' \subset W$, no witness has a vertex in its
  % quadrant~IV, then $\RG^+(P,W')$ is a staircase graph of type~IV.
  In a witness rectangle graph, if no witness has a vertex in its quadrant~IV, then the
  graph is a staircase graph of type~IV, possibly after removing some
  redundant witnesses.
\end{lemma}
\begin{proof}
  First remove any redundant witnesses, if present.  Now apply
  Lemma~\ref{quadEmpty}, to the possibly smaller, new witness set, to
  conclude that every remaining witness has its quadrant~IV empty of
  witnesses as well.  Therefore the remaining witnesses form an
  ascending chain and all vertices lie above it, as in the definition
  of a staircase graph of type~IV.  \hfill $\Box$
\end{proof}

% \begin{wrapfigure}[10]{R}{0pt}
%   \includegraphics[width=4cm]{RGparallel1}
%   \caption{The case with all edges of positive slope.}
%   \label{RGparallel1}
% \end{wrapfigure}
Of course, if the empty quadrant in the above lemma is not~IV but I,
II, or~III, we get a staircase graph of the corresponding type.

\begin{theorem}
  \label{staircase}
  In a witness rectangle graph with two non-trivial connected
  components, each component is isomorphic to a staircase graph.
  Conversely, the disjoint union of two graphs representable as
  staircase graphs is isomorphic to a witness rectangle graph.
\end{theorem}

\begin{proof}
  We distinguish two cases.

%\paragraph*{All edges have the same slope:}
\noindent\emph{All edges have the same slope:}
Suppose all edges of components $C_1$, $C_2$ have the same slope, say
positive.  Let $ab$ be an edge of $C_1$ witnessed by $w_1$ and $cd$ be
an edge in $C_2$ witnessed by $w_2$, with $x(a)<x(b)$ and $x(c)<x(d)$.
By Lemma~\ref{lem:2sameSlope}, $w_1$ and $w_2$ are distinct.  The
vertices $a$ and $b$ are in quadrants III~and~I of $w_1$,
respectively.  As $c$ and $d$ are not adjacent to $a$ or $b$, and as
$cd$ doesn't share its witness with $ab$, $c$ and $d$ are both in
quadrant~II or both in quadrant~IV of $w_1$.  Suppose, without loss of
generality, that $cd$ is in quadrant~IV of $w_1$ (see
%  figure~\ref{RGparallel1}).
  Figure~\ref{RGcrossing1}).
% \begin{figure}
% \centering
% \includegraphics[width=4cm]{RGparallel1}
% \caption{The case of all edge of positive slope.}
% \label{RGparallel1}
% \end{figure}
By a symmetric argument, $ab$ is in quadrant~II of $w_2$.  This
holds for any two edges $ab \in C_1$ and $cd \in C_2$.  (Given two
edges $ef \in C_1$ and $gh \in C_2$, we say $ef < gh$ if a witness of
$ef$ is in quadrant~II of a witness of $gh$, and $gh < ef$ otherwise.
We claim that either $ef < gh$ for all choices of edges $ef \in C_1$ and
$gh \in C_2$, or $gh < ef$, for all such choices.  Otherwise there
would have to exist, without loss of generality, a triple of edges $ef,
ij \in C_1$ and $gh \in C_2$, with $ef < gh < ij$.  This implies that
some witnesses $w'$, $w''$, $w''$ of $ef$, $gh$, $ij$, respectively,
form a descending chain.  Considering the relative positioning of the
three edges and three witnesses, we conclude that $ef$ and $ij$ must be
independent.   Hence we have three pairwise independent edges in
$E^+$, contradicting Lemma~\ref{lem:no3independentSameSlope}, 
thereby proving the claim.)
Notice that no vertex of $C_1$ is in quadrant~II of any witness of
$C_1$ or it would be connected to $C_2$.  Similarly, no vertex of
$C_2$ is in quadrant~IV of any witness of $C_2$ or it would be
connected to $C_1$.  Hence, by Lemma~\ref{quadEmpty2}, $C_1$ and
$C_2$ are both staircase graphs.

%\paragraph*{At least two edges have opposite slopes:}
\noindent\emph{At least two edges have opposite slopes:}
There is at least one pair of edges $ab \in C_1$ and $cd \in C_2$ of opposite slopes.
Suppose, without loss of generality, that $ab \in E^+$ and $cd \in E^-$.
By Lemma~\ref{lem:2oppositeSlope}, $ab$ and $cd$ share a witness $w$
(see Figure~\ref{RGcrossing1}, center).
\begin{figure}[t]
  \centering
  \includegraphics[width=0.32\linewidth]{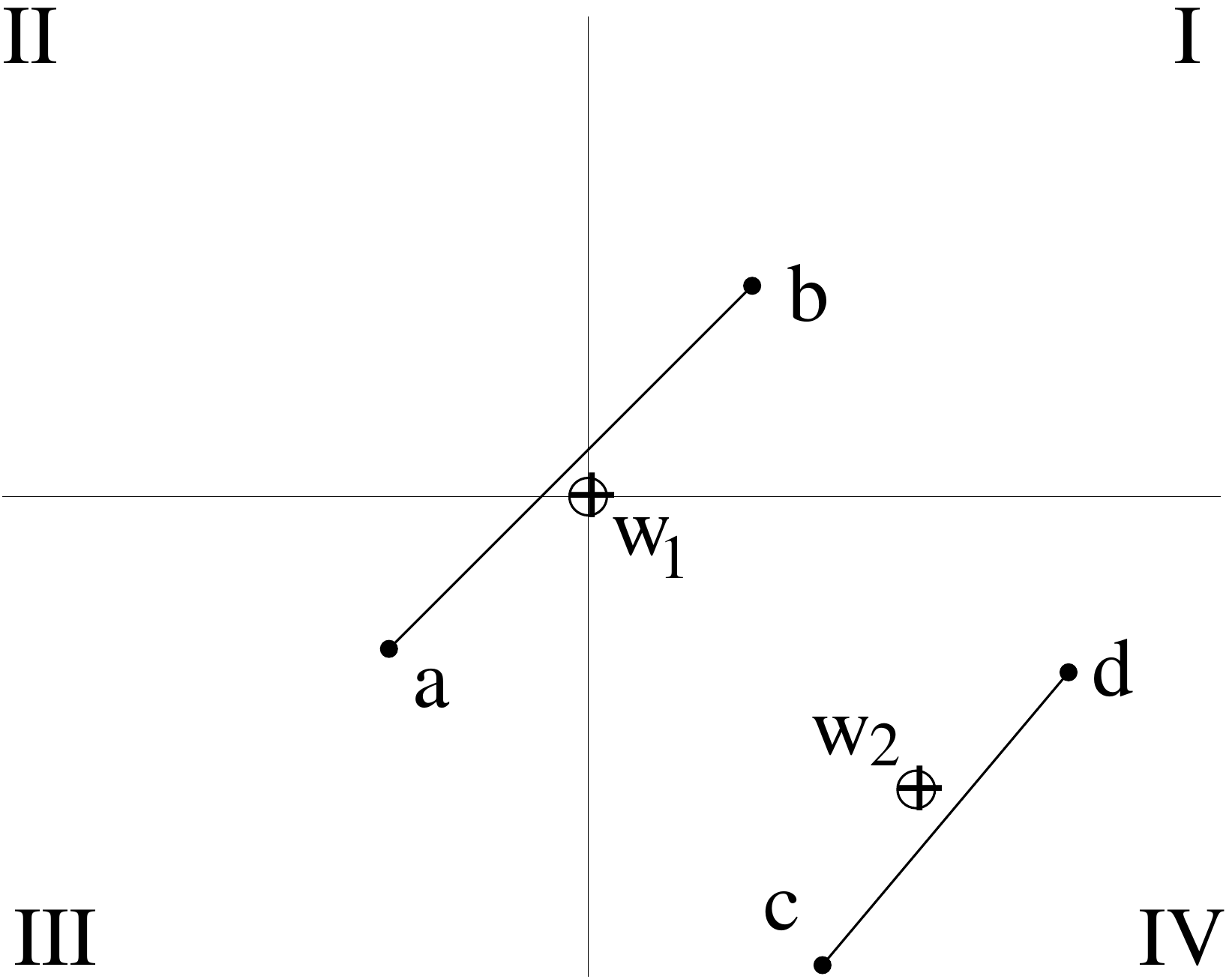}%
  \hspace*{0.03\linewidth}%
  \includegraphics[width=0.32\linewidth]{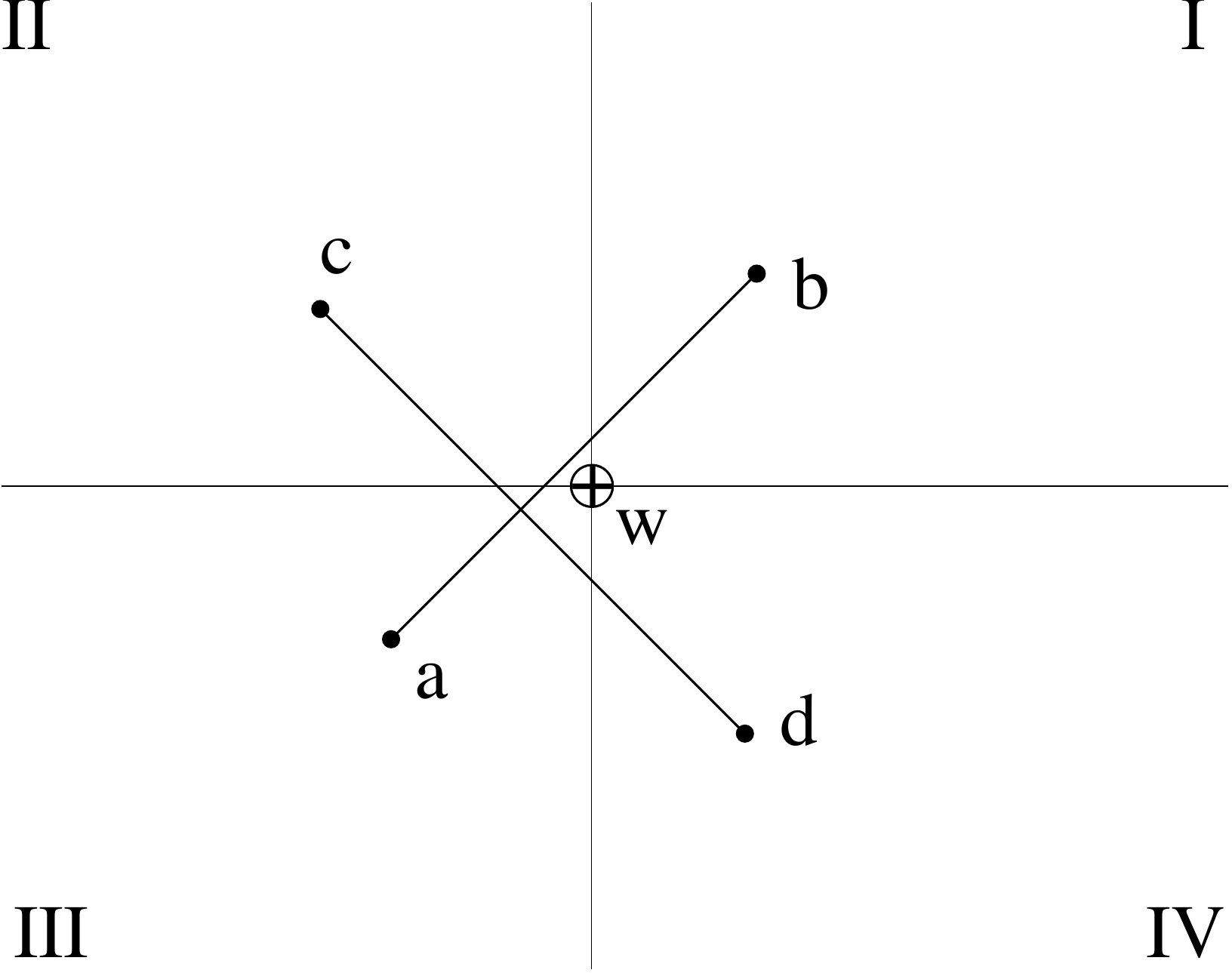}%
  \hspace*{0.03\linewidth}%
  \includegraphics[width=0.32\linewidth]{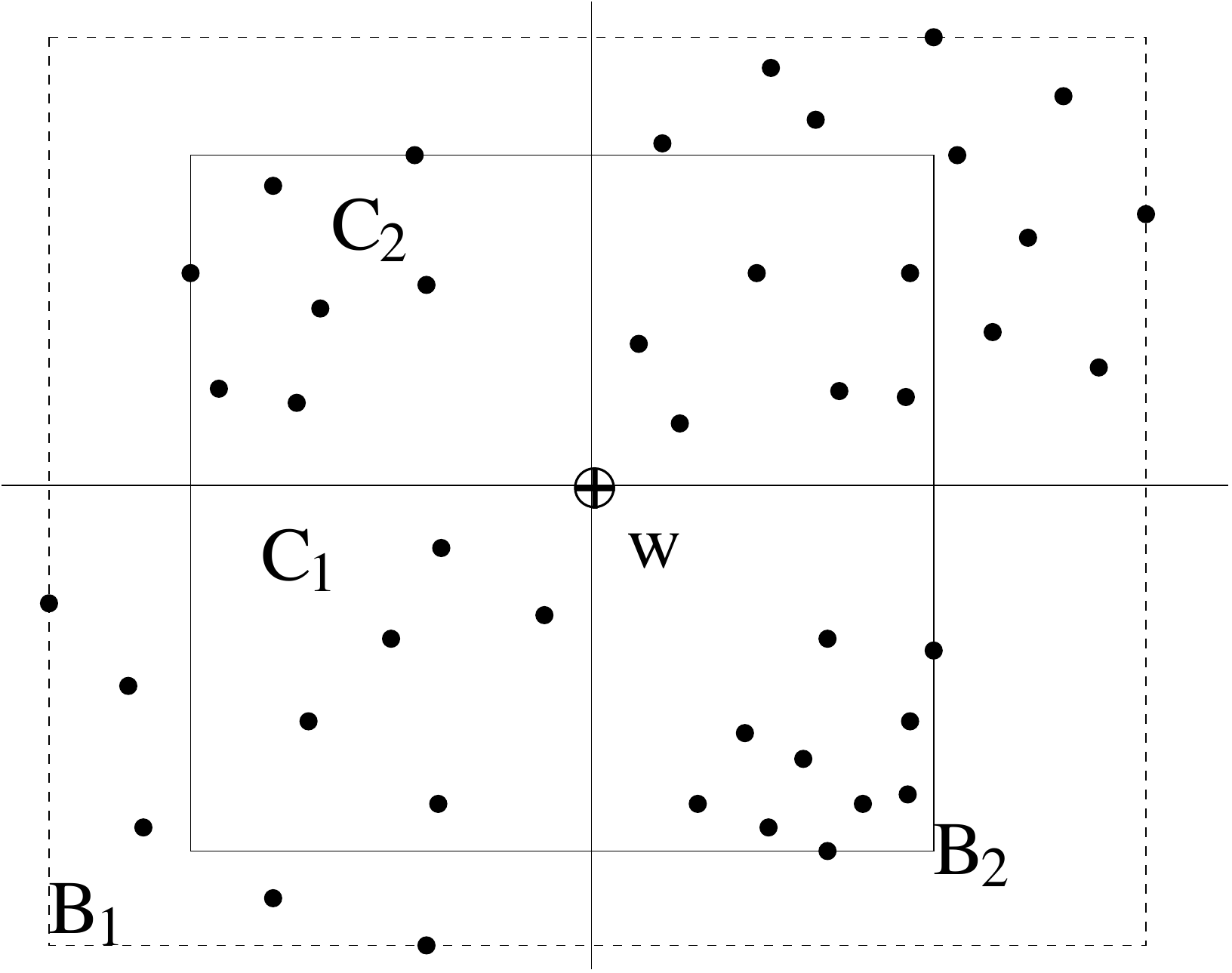}
  \caption{Cases in the proofs of Lemmas~\ref{lem:2sameSlope},
    \ref{lem:2oppositeSlope},
    \ref{lem:no3independentSameSlope}, \ref{lem:no3independent} and
    Theorem~\ref{staircase}: Two edges of positive slope (left), two
    edges of opposite slopes (center). Minimum bounding boxes of $V(C_1)$: $B_1$ (dashed),
    and of $V(C_2)$: $B_2$ (plain). $V(C_1)$ is in quadrants I~and~III, $V(C_2)$ is in quadrants II~and~IV (right).}
    \label{RGcrossing1}
\end{figure}

Draw isothetic boxes $B_1$ and $B_2$, as follows:
$B_1$ is the minimum bounding box of the vertices of $C_1$, while
$B_2$ is its analog for $C_2$ (see Figure~\ref{RGcrossing1}, right).
%\begin{figure}
%  \centering
% \includegraphics[width=0.4\linewidth]{BRGboxes2}
%  \caption{Minimum bounding boxes of $C_1$: $B_1$ (dashed),
%    and of $C_2$: $B_2$ (plain). $C_1$ is in quadrants I~and~III, $C_2$ is in quadrants II~and~IV. }
%  \label{RGcrossing2}
%\end{figure}

Consider $C_1$; we argue that it is isomorphic to a staircase graph;
as we can apply the same reasoning to $C_2$, this will imply the first
part of the 
theorem.  By Lemma~\ref{lem:2oppositeSlope}, every edge in $E^+ \cap C_1$ shares some witness with $cd$;
the witness must therefore lie in $B:=B_1 \cap B_2$.  Let $W'$ be the
set of all such witnesses; $w \in W'$.  All vertices of $C_1$ lie in
quadrants I~and~III of every $w' \in W'$, and all vertices of $C_2$
lie in quadrants II~and~IV of every $w'$; otherwise $C_1$ and $C_2$
would be connected.  Now remove redundant witnesses from $W'$, i.e.,
pick a minimal subset $W''$ of $W$ such that $H:=\RG^+(V(C_1),W')$
coincides with $\RG^+(V(C_1),W'')$.  $H$ is a staircase graph and the
witnesses of $W''$ form an ascending chain, by Lemma~\ref{quadEmpty2}.
% (Moreover, as all witnesses in $W''$ have their quadrants II~and~IV
% empty of vertices of $C_1$, vertices of $H$ are naturally partitioned
% into sets confined between consecutive witnesses of $W''$ and in fact
% $H$ is a join of several graphs with no edges; refer to
% Figure\Boris{DOES NOT EXIST; DO WE NEED IT?}.)

Consider the portion of $C_1$ in quadrant~III of $w$.  All witnesses
to the left of $B_2$ and above its lower edge, or below $B_2$ and to
the right of its left edge, have two consecutive quadrants empty of
vertices of $C_1$ (if there were some vertices of $C_1$ in these
quadrants they would be adjacent to vertices of $C_2$ defining the
borders of $B_2$).  Therefore such a witness would not witness any edges
and cannot be present in $W$.  %\Boris[thinks]{Need a figure?}

% Hence all remaining witnesses (in quadrant~III of $w$) of edges of $C_1$
% must lie below and to
% the left of $B_2$ (and, of course, in $B_1$); let $W'''$ be a minimal
% subset of such witnesses inducing the same edges.
Hence all remaining witnesses (in quadrant~III of $w$) of edges of $C_1$
must lie below and to the left of $B_2$ (and, of course, in $B_1$).  All
these witnesses have $B_2$ in their first quadrant, therefore they
witness edges of $E^-\cap E(C_1)$.  Let $W'''$ be a minimal subset of such
witnesses.
Each of them has
their quadrant~III empty of vertices of $C_1$ (any vertex of $C_1$ in
quadrant~III of such a witness $w'''$ would be adjacent to vertices in
$C_2$ as $B_2$ is in quadrant~I of $w'''$), therefore by
Lemma~\ref{quadEmpty2}, $\RG^+(V(C_1),W''')$ is a staircase graph.

Consider the lowest leftmost witness $w_\ell$ of $W''$ (recall that
they form an ascending chain).  As shown previously, $w_\ell$ doesn't
have any vertex of $C_1$ in its second and fourth quadrants.  Let
$V_1$ be the set of vertices of $C_1$ in its first quadrant and let
$V_2$ be the vertices of $C_1$ in its third quadrant.
As shown above, $\RG^+(V_1,W'' \setminus \{w'_\ell \})$ and
$\RG^+(V_2,W''')$ are staircase graphs, and, therefore, by
Lemma~\ref{lem:staircase-splits}, 
$\RG^+(V(C_1), W'' \cup W''')$ is isomorphic to a staircase graph.
We apply similar reasoning to quadrant~I of $B$, to conclude that
$C_1$ is a join of at most three staircase graphs and therefore
isomorphic to a staircase graph.

$C_2$ is handled by a symmetric argument (in fact, though it does not
affect the reasoning as presented, either $C_1$ contains
negative-slope edges, or $C_2$ contains positive-slope edges, but not
both), concluding the first part of the proof.

%\paragraph*{Converse:} 
\noindent\emph{Converse:} 
Given two staircase graphs $G_1$ and $G_2$, place a scaled and
reflected copy of $G_1$ in quadrant~I of the plane, with witnesses on
the line $x+y=1$ and vertices below the line.  Place a scaled and
reflected copy of $G_2$ in quadrant~III of the plane, with witnesses
on the line $x+y=-1$ and vertices above the line.  It is easy to check
that the result is a witness rectangle graph isomorphic to $G_1 \cup G_2$. 

 \hfill $\Box$
%, as desired.
\end{proof}

\section{What Are Staircase Graphs, Really?}
\label{sec:what-are-staircase-graphs}

The above discussion is %quite 
unsatisfactory in that it describes one
new class of graphs in terms of another such new class.  In this
section, we discover that the class of graphs representable as
staircase graphs is really a well known family of graphs.

Recall that an \emph{interval graph} is one that can be realized as
the intersection graph of a set of intervals on a line, i.e., its set
of vertices can be put in one-to-one correspondence with a set of
intervals, with two vertices being adjacent if and only if the
corresponding intervals intersect.  A \emph{co-interval graph} is the
complement of an interval graph, i.e., a graph representable as a
collection of intervals in which adjacent vertices correspond to
disjoint intervals.

\begin{lemma}
  \label{lem:staircase-interval}
  Graphs representable as staircase graphs are exactly the co-interval
  graphs.
\end{lemma}

\begin{proof}
  Consider a vertex $v$ in a staircase graph $\RG^+(V,W)$.  Without
  loss of generality, assume the witnesses~$W$ lie on the line $\ell
  \colon y=x$ and the vertices lie above it.  %\Boris{This used to say:
%    ``Associate $v$ with the smallest interval $I_v$ of $\ell$
%    containing all witnesses lying to the right and below $v$ as well
%    as the witness immediately above $v$ and the witness immediately
%    to the left of $v$.''  I think that does not work when the witness
%    just above one vertex is just to the right of another---the
%    intervals intersect in that witness, but the vertices are
%    adjacent, due to the witness.  What follows is an attempted fix.
%    CHECK!}
    Create an artificial witness on $\ell$ lying above all
  vertices.  Associate $v$ with the smallest interval $I_v$ of $\ell$
  containing all witnesses lying to the right and below $v$ as well as
  the witness immediately above $v$.  It is easily checked that $v
  \sim v'$ in $\RG^+(V,W)$ if and only if $I_v$ and $I_{v'}$ do not
  meet.  Thus the intersection graph of the intervals $\{ I_v \mid v
  \in V\}$ is isomorphic to the complement of $\RG^+(V,W)$.

  Conversely, let $H$ be a co-interval graph on $n$ vertices and
  %let 
  $\{I_v\}$ %be 
  its realization by a set of closed intervals on the
  line $\ell \colon x=y$.  Extend each interval, if necessary,
  slightly, to ensure that the $2n$ endpoints of the intervals are
  distinct.  Place $2n-1$ witnesses along $\ell$, separating
  consecutive endpoints, and transform each interval
  $I_v=[(a_v,a_v),(b_v,b_v)]$ into point $p_v=(a_v,b_v)$.  Let $W$ and
  $P$ be the set of witnesses and points thus generated.  Now $I_v$
  misses $I_w$ if and only if $a_v < b_v < a_w < b_w$ or $a_w < b_w <
  a_v < b_v$, which happens %precisely 
  if and only if the rectangle
  $B(p_v,p_w)$ contains a witness.  Hence $\RG^+(P,W)$ is isomorphic
  to $H$, as claimed.  \hfill $\Box$
\end{proof}

Theorem~\ref{staircase} and Lemma~\ref{lem:staircase-interval} imply
the following %much 
more satisfactory statement:
%\Boris{A bit awkward.  Is there a prettier phrasing?}
\begin{theorem}
  \label{thm:staircase-two-intervals}
  The class of graphs representable as witness rectangle graphs with
  two non-trivial connected components is precisely the class of
  graphs formed as the disjoint union of zero or more isolated vertices
  and two co-interval graphs.
\end{theorem}

\begin{corollary}
  Whether or not a given combinatorial graph $G=(V,E)$ with two
  non-trivial connected components can be drawn as a witness rectangle graph can be tested
  in time $O(|V|+|E|)$; a drawing, if it exists, can be constructed in the same time.
  % A representation can be constructed within
  % the same time bounds, if one exists.  \boris{I cannot tell if the
  %   last statement is true from the discussion in \cite{IGraph}---they
  %   focus on recognition; this is discussed on section 5 there;
  %   co-interval graphs are at the end of the section.  It could be
  %   that their earlier $O(n+m\log n)$ algorithm can be used for
  %   reconstruction and this one cannot be used for reconstruction but
  %   only of recognition.  No idea.  Would be nice to ask someone who
  %   knowns.}
\end{corollary}

\begin{proof}
  Use the linear-time recognition and reconstruction algorithm for co-interval graphs from~\cite{IGraph,PComm}.
  \hspace*{0pt} \hfill $\Box$
  % \Boris{reconstruction? or drop the claim in GD
  %   submission.}\ferran{Let's say recognition, I have had no time to check the paper
  %   you quote and even the concept is unclear, at some point I will check the paper
  %   I mention in my mail of June 4.}
\end{proof}

\section{Domination number}
For a graph $G$ and a subset $S$ of the vertex set  $V(G)$, denote by $N_G[S]$ the set of vertices in $G$  which are in $S$ or adjacent to a vertex in $S$. If $N_G[S]=V(G)$, then $S$  is said to be a \emph{dominating set} of vertices in  $G$.
The \emph{domination number} of a graph $G(V,E)$, denoted $\gamma(G)$, is the minimum size of a dominating set of vertices in $V$ \cite{mathworld:DominatingSet, mathworld:DominationNumber}.

\begin{lemma}
\label{DominationNumberStaircase}
The maximum domination number of a connected staircase graph is two.
\end{lemma}
\begin{proof}
For the lower bound, consider a staircase graph of type $IV$ that is a path of length three. It contains four vertices.
It is easy to check that it cannot be dominated by one vertex.
%Pick any vertex $d$, it can not be a dominating set as there is at least one vertex $v \neq d$ that is not adjacent to $d$.
%Therefore the domination number is at least two.

For the upper bound, consider a non-trivial connected staircase graph of type~$IV$ (see Figure~\ref{DominationStaircaise2}). 
%\begin{wrapfigure}[10]{l}{0.55 \textwidth}
%\scalebox{0.5}{\input{DominationStaircaise2.pdf_t}}
% \caption{The two bold vertices form a dominating set for this staircase graph.}
% \label{DominationStaircaise2}
%\end{wrapfigure} 
\begin{figure}
\centering
\scalebox{0.5}{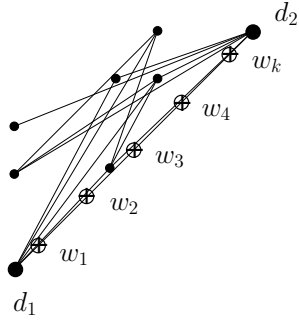}
   \caption{The two bold vertices form a dominating set for this staircase graph.}
   \label{DominationStaircaise2}
\end{figure}
Let $w_i$, $1 \leq i \leq k$, be the witnesses in order from left to right, with $w_1$ being the bottommost leftmost witness and $w_k$ the uppermost rightmost witness.
Recall that we assume that the set of witnesses is minimal, hence there is no witness with two adjacent quadrants empty of vertices.
This implies that there is at least one vertex $d_1$ below $w_1$ and at least one vertex $d_2$ to the right of $w_k$.
The vertices $d_1$ and $d_2$ form a dominating set of size two.
Indeed all vertices of the graph are in quadrants $I$, $II$ and $III$ of $w_k$.
The vertex $d_1$ is adjacent to all vertices in the graph that are in quadrant $I$ and $II$ of $w_k$, and $d_2$ is adjacent to all vertices in quadrant $III$ of $w_k$ (since the graph is connected, there are no vertices lying to the left and above of both $w_1$ and $w_k$ -- such vertices would have to be isolated.)
 \hfill $\Box$
\end{proof}

\begin{theorem}
\label{DominationNumberWRG}
The maximum domination number of a witness rectangle graph with non-trivial connected components is four.
\end{theorem}
\begin{proof}
We will consider three cases:

\begin{enumerate}[i)]
\item%[\emph{First case:}] 
There is a witness $w$ in WRG such that there is at least one vertex in each quadrant of $w$. Pick one vertex in each quadrant, they form a dominating set of size four.
Indeed, the vertices in quadrant $I$, $II$, $III$, and $IV$ are respectively adjacent to all vertices in quadrant $III$, $IV$, $I$, and $II$.

\item%[\emph{Second case:}] 
There is a witness $w$ in WRG such that it has two opposite quadrants with at least one vertex in each of them, and two opposite quadrants empty of vertices.
Pick one vertex in each of the non-empty quadrants, they form a dominating set of size two.
Indeed suppose without loss of generality that the non-empty quadrants are~$I$ and~$III$.
Any vertex in quadrant~$I$ is adjacent to all vertices in quadrant~$III$, and any vertex in quadrant~$III$ is adjacent to all vertices in quadrant~$I$.

\item%[\emph{Last case:}] 
All witnesses have one quadrant empty of vertices and three adjacent quadrants containing at least one vertex each.
This is the last case as the set of witnesses is minimal which implies that there is no witness with vertices in one quadrant only or in two adjacent quadrants only.

Let the witnesses with one quadrant empty be the witnesses of type~$I$, $II$, $III$, and $IV$ with the number of the type of witness corresponding to the number of the empty quadrant. 
For each type of witness, all witnesses of this type witness the edges of a staircase graph of this type, which is a subgraph of the WRG.
The union of the four staircase graphs defined by the four types of witnesses is the original WRG (see Figure~\ref{DominationNumber}, for clarity the edges are not represented).
\begin{figure}
\centering
\scalebox{0.6}{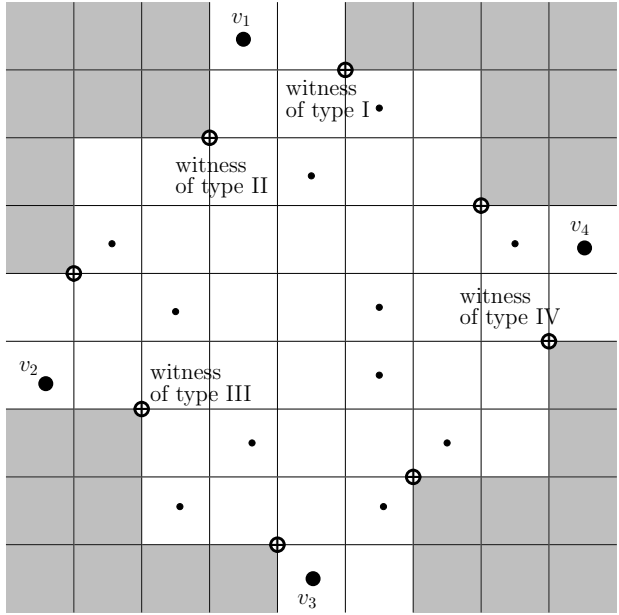}
\caption{WRG graph with all witnesses having exactly one quadrant empty. The bold black vertices are the dominating set.}
\label{DominationNumber}
\end{figure}
%\boris{Does this need more elaboration?}

By definition of a staircase graph, all its witnesses form a monotone ascending (for staircase graphs of type~$II$ and~$IV$) or descending (for staircase graphs of type~$I$ and~$III$) chain.
Moreover notice that all witnesses of type~$II$ must be in quadrant~$II$ or~$III$ of all witnesses of type~$I$ , all witnesses of type~$III$ must be in quadrant~$III$ or~$IV$ of all witnesses of type~$II$, all witnesses of type~$IV$ must be in quadrant~$IV$ or~$I$ of all witnesses of type $III$, and all witnesses of type $I$ must be in quadrant~$I$ or $II$ of all witnesses of type $IV$.
Reciprocally all witnesses of type~$I$ must be in quadrant~$I$ or~$IV$ of all witnesses of type~$II$, all witnesses of type~$II$ must be in quadrant~$I$ or~$II$ of all witnesses of type~$III$, all witnesses of type~$III$ must be in quadrant~$II$ or~$III$ of all witnesses of type $IV$, and all witnesses of type $IV$ must be in quadrant~$III$ or $IV$ of all witnesses of type $I$.
If one of these constraints is violated, one witness would have two adjacent quadrants that are empty, and therefore the set of witnesses would not be minimal, a contradiction.

From this set of constraints it follows that either the vertices in quadrant $II$ of the topmost witness of type $I$ are included in quadrant $I$ of the topmost witness of type $II$ and/or the vertices in quadrant $I$ of the topmost witness of type $II$ are included in quadrant $II$ of the topmost witness of type $I$. Hence there is a vertex $v_1$ in the intersection of quadrant $II$ of the topmost witness of type $I$ and quadrant $I$ of the topmost witness of type $II$. Similarly there is a vertex $v_2$ in the intersection of quadrant $III$ of the leftmost witness of type $II$ and quadrant $II$ of the leftmost witness of type $III$; there is a vertex $v_3$ in the intersection of quadrant $IV$ of the bottommost witness of type $III$ and quadrant $III$ of the bottommost witness of type $IV$; and finally there is a vertex $v_4$ in the intersection of quadrant $I$ of the rightmost witness of type $IV$ and quadrant $III$ of the rightmost witness of type $I$. Following the proof of Lemma~\ref{DominationNumberStaircase}, $v_1$ and $v_2$ form a dominating set for the staircase graph of type~$II$, $v_2, v_3$ form a dominating set for the staircase graph of type~$III$, $v_3, v_4$ form a dominating set for the staircase graph of type~$IV$, and $v_4, v_1$ form a dominating set for the staircase graph of type~$I$.
Therefore $v_1,v_2,v_3,v_4$ form a dominating set for WRG.

If there are witnesses of fewer than four types, four vertices are also sufficient.
Indeed as shown in the proof of Lemma~\ref{DominationNumberStaircase}, each staircase graph requires two vertices in its dominating set, and as we have seen above, staircase graphs of consecutive types share a dominating vertex.
\end{enumerate}
 \hspace*{0pt} \hfill $\Box$
\end{proof}

\section{Mutual witness graphs and set separability}
\label{sec:MNG}
%\Boris{I think we need a different section name.  Ferran?}
In~\cite{MNG}, Manabu Ichino and Jack Sklansky introduce
``\emph{mutual neighborhood graph}'' defined by a pair of point sets $(P,W)$
in the plane.  This concept coincides with that of the negative
witness rectangle graph $\RG^-(P,W)$ in which there is an edge between
two points $p,q$ in $P$ if and only if $B(p,q)$ does not contain any witness $w
\in W$.  
For a fixed pair of sets $A,B$, they focus on the relation between the two negative witness rectangle graphs $\RG^-(A,B)$ and $\RG^-(B,A)$.
In particular, they claim that if $\RG^-(A,B)$ and $\RG^-(B,A)$ are complete graphs, then the pair $(A,B)$ is linearly separable.
Unfortunately this is not true, as demonstrated by the counter-example in Figure~\ref{MNG1}.
\begin{figure}
\centering
\includegraphics[scale=0.7]{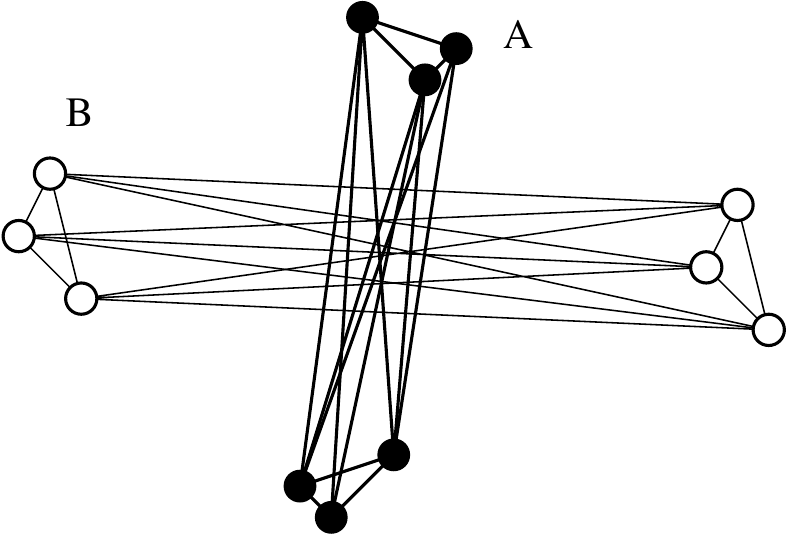}
\caption{$\RG^-(A,B)$ and $\RG^-(B,A)$ are complete (say $A$ is the set of black points and $B$ the set of white points) but the pair $(A,B)$ is not linearly separable.}
\label{MNG1}
\end{figure}

In~\cite{thesis}, the authors have studied some similar properties of corresponding witness proximity graphs in which one set is the set of vertices, the other set is the set of witnesses and inversely, and they prove that for the negative witness Delaunay graph ($\DG$, introduced in~\cite{ADH}), given two disjoint sets of points $A$ and $B$, if $\DG^-(A,B)$ and $\DG^-(B,A)$ are complete, the pair $(A,B)$ is circularly separable, but not necessarily linearly separable.
On the contrary, for the negative witness Gabriel graph ($\GG^-(P,W)$, introduced in~\cite{ADH08}), given two disjoint sets of points $A$ and $B$, if $\GG^-(A,B)$ and $\GG^-(B,A)$ are complete, the pair $(A,B)$ is linearly separable.

\section{Combining positive and negative witnesses}
\label{sec:PlusMinusW}
In this section we define a variant of witness rectangle graphs
that combine positive and negative witnesses. 
% \Boris{should be have a
%  name for these graphs? Some of the text below would read better if
%  they had a name.}  
  Given a set of vertices $P$ in the plane, a
set $W^+$ of positive, and a set $W^-$ of negative witnesses, we
define $\RG^\pm(P,W^+,W^-)$ as follows.  We assume that $P \cup W^+
\cup W^-$ has no two points with the same $x$-coordinates and no two
points with the same $y$-coordinates.  Define $S(p,q):=|B(p,q) \cap
W^+|-|B(p,q) \cap W^-|$.  The graph has points of $P$ as vertices and $p
\sim q$ if and only if $S(p,q)>0$.  In words, $p$
is adjacent to $q$ whenever the rectangle defined by them has more
positive witnesses than negative ones.

This class of graphs is a proper generalization of both the positive 
and negative witness rectangle graphs, for an appropriate
choice of the sets $W^+$ and $W^-$.  
%\Boris{A subtlety with emulating
%  the negative version: Can it be done with approximately the same
%  number of witnesses?} 
   In fact, \emph{any} graph can be realized:
% We show below that,
% without restricting the number of witnesses, these graphs are ``too
% flexible'':

\begin{theorem}
  Any combinatorial graph $G$ on $n$~vertices can be drawn as a witness
  rectangle graph with positive and negative witnesses using at most
  $(n-1)^2$~witnesses.
\end{theorem}

\begin{proof}
  Consider an $n \times n$ integer grid.  Identify the vertices of $G$
  with the points $p_i=(i,i)$ on the diagonal, for $i=1, \ldots, n$,
  in an arbitrary but fixed order (see Figure~\ref{WRG+-}).

  We never put witnesses on the grid lines.  We will place witnesses
  in the (open) squares of the grid, maintaining the following
  properties: (a)~each square contains at most one (positive or negative) witness, (b)~for
  any $p,q \in V(G)$, $S(p,q)$ is either zero or one, and
  (c)~$S(p,q)=1$ if and only if $p \sim q$.

  The process begins with $W^+=W^-=\emptyset$.  We start with the
  diagonal containing vertices and work our way outward.  We maintain
  the following invariant: property~(a) holds all the time, and, at
  step~$k$ properties (b)~and~(c) hold for all pairs $p_i,p_j$ with
  $|i-j|\leq k$.  When $k$ reaches $n-1$, we are done.

  \begin{itemize}
  \item[\emph{Base case ($k=1$)}:] For every pair of consecutive
    vertices $p_i$, $p_{i+1}$, we put a positive witness into
    $B(p_i,p_{i+1})$ if $p_i\sim p_{i+1}$ and place no witness
    otherwise.  The invariant is clearly satisfied (see Figure~\ref{WRG+-} left).
  \item[\emph{Inductive step ($k>1$)}:] Consider a pair of vertices
    $p_i,p_{i+k}$.  Let $X:=B((i,i+k-1),(i+1,i+k))$ 
    and $Y:=B((i+k-1,i-k),(i+k,i-k+1k))$.  We write $S_X$
    for $|X \cap W^+| - |X \cap W^-|$ and define $S_Y$ similarly.
    Since there are no witnesses on the grid lines, by
    inclusion-exclusion principle we have $S(p_i,p_{i+k})=S_X + S_Y +
    S(p_{i+1},p_{i+k})+ S(p_i,p_{i+k-1}) - S(p_{i+1},p_{i+k-1})$.  By
    the inductive hypothesis, each of the numbers $S(p_{i+1},p_k),
    S(p_i,p_{i+k-1}), S(p_{i+1},p_{i+k-1})$ is either zero or one.  In
    particular, the last three terms sum up to an integer between $-1$
    and $2$.  To maintain the invariant, we need $S(p,q)$ to be $0$ or
    $1$, depending on whether $p \sim q$.  Hence the combined
    contribution of $X$ and $Y$ has to be between $-2$ and $+2$.
    However, $S_X+S_Y$ can be arranged to have any value from
    $-2,-1,0,1,2$ by placing one witness of an appropriate sign, or no
    witness, in $X$ and in $Y$ (see Figure~\ref{WRG+-} right).

    Repeating this process for all pairs $(p_i,p_{i+k})$ and noticing
    that the each pair of vertices uses a pair of different boxes
    $(X,Y)$, we complete the inductive step: Every square of the grid
    still contains at most one witness, and for all pairs of vertices
    $(p_i,p_j)$ up to $k$ apart, the box $B(p_i,p_j)$ contains an
    equal number of positive and negative witnesses if $p_i \not\sim
    p_j$, and positive witnesses outnumber negative ones by one,
    otherwise.  \hfill $\Box$
  \end{itemize}
\end{proof}

%\begin{theorem}
%Every set of vertices in the plane can be used for realizing any combinatorial graph as a witness rectangle graph with positive and negative witnesses.
%\end{theorem}
%\begin{proof}
%Consider an arbitrary combinatorial graph with $n$ vertices.
%Place $n$ points arbitrary in the plane, and create a bijection between the $n$ vertices from the combinatorial graph and the $n$ points in the plane.
%Let the points in the plane be $p_1,\ldots, p_n$ in order of increasing abscissa.
%Realize inductively the subgraph induced by $p_2,\ldots,p_n$.
%For the adjacencies involving $p_1$, add witnesses inside the vertical strip bounded by
%$p_1$ and $p_2$.
%For the points with greater y than $p_1$, do this from $p_1$ up, in order of increasing ordinate
%correcting in every new box the sign if necessasry.
%For the points with smalleer y than $p_1$, do this from $p_1$ down, in order of decreasing ordinate.
%  \hfill $\Box$
%\end{proof}

%\Boris[suggests]{Muriel, Could you please adjust the figures to math
%  the text and cite them in the right places? We need $p_i$, $p_{i+1}$
%  in base case figure and $p_i$, $p_{i+k}$, $X$, and $Y$ in the second
%  figure, maybe with the for larger boxes highlighted?}
%Figure~\ref{WRG+-}(left) and (right)
\begin{figure}
\centering
\scalebox{0.5}{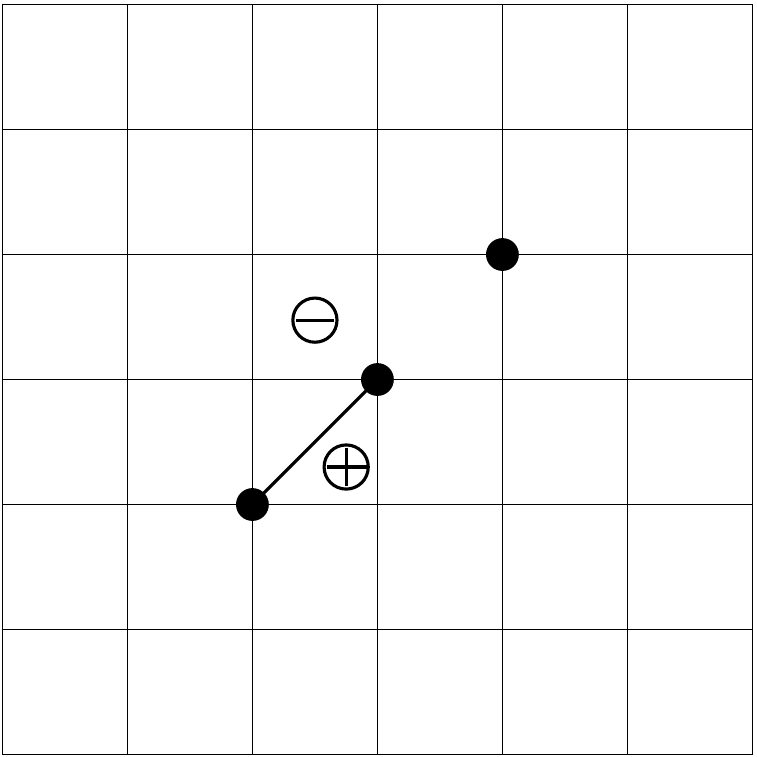}
\hspace{0.2\linewidth}%
  \scalebox{0.5}{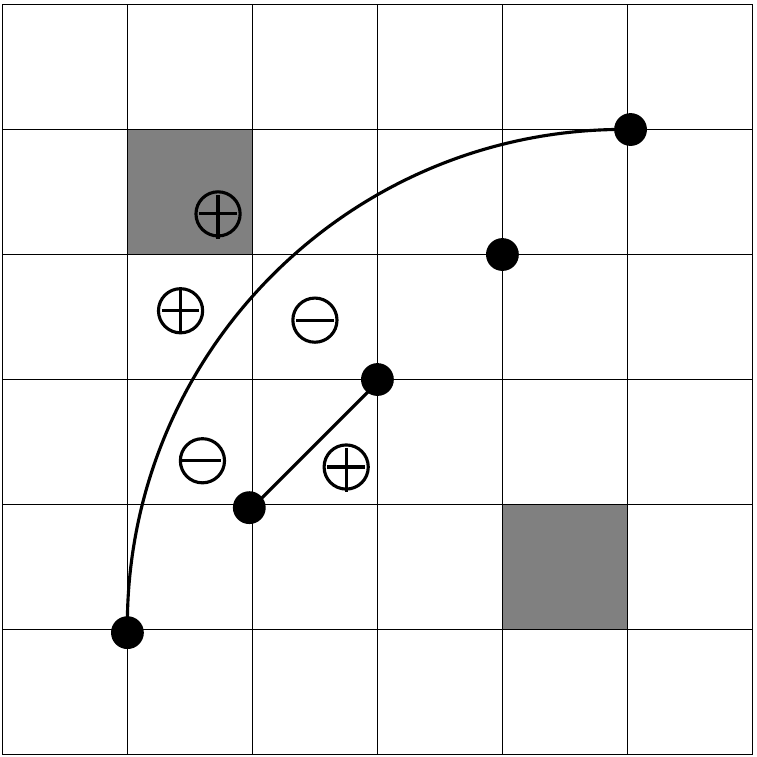}
  \caption{Vertices are the black disks, positive witnesses are the white disks with a cross, and negative witnesses are the white disks with a minus sign.}
  \label{WRG+-}
\end{figure}

The proof suggests a number of open problems: 
Is it possible that a
significantly smaller number of witnesses is sufficient to realize all
possible combinatorial graphs? We conjecture that a quadratic number
of witnesses is necessary in the worst case. 
A related open problem is to find a ``good'' set of vertices to realize a given graph, as it can be proved that every graph can be realized
on top of any point set whose elements have no duplicated $x$ or $y$
coordinates, and the number of witnesses depends on the specific initial set.
Another open problem is to find classes of graphs
that are realizable with much fewer witnesses, say, $O(n)$ or $O(1)$.
Also attaching real-number weights to witnesses
may possibly allow one to realize graphs with fewer witnesses.

% %%%%%%%%%%%%%%%%%%%%%%%%%%%%%%%%%%%%%%%%%%%%%%%%%%%%%%%%%%%%%%%%%%%%%%%%%%%%%%
\section{How to Stab Rectangles, Thriftily}\label{section:stabbing}

Let $P$ be a set of $n$ points in the plane, and let $S$ be some given
family of geometric regions, each with at least two points from $P$ on
its boundary.  The problem of how many points are required to stab all
the elements of $S$ using a second set $W$ of points has been
considered several times for different families of
regions~\cite{ADH,ADH08,KM88,CKU99,SU07}. For example, among the
shapes previously investigated were the family of triangles with
vertices in $S$ and the family of disks whose boundary passes through
two points of $P$.
This family of problems has a natural formulation in terms of appropriate witness graphs.

We consider here a variant of this problem that is related to WRGs, in
which we focus on the family $R$ of all open isothetic rectangles
containing two points of $P$ on their boundary and assume that the
points of $P$ have no repeated $x$- or $y$-coordinates.  We denote by
$st_R(n)$ the maximum number of piercing points that are required,
when all sets $P$ of $n$ points are considered.  Stabbing all the
rectangles that have $p$ and $q$ on their boundary is equivalent to
just stabbing $B(p,q)$.  Therefore we see that
% \[st_R(n)=\max_{|P|=n} st_R(P)=\max_{|P|=n} \min \{|W| : \RG^+(P,W)=K_{|P|}\}.\]
\[st_R(n)=\max_{|P|=n} \min \{|W| : \RG^+(P,W)=K_{|P|}\}.\]

\begin{theorem} \label{thm:stabbingRectangles}
  The asymptotic value of $st_R(n)$ is $2n-\Theta(\sqrt{n})$.
% BA Ferran had: The asymptotic value of $st_R(n)$ is $st_R(n)=2n-\Theta(\sqrt{n})$.
%  \muriel{The reviewers want to know the constant in $\Theta$}
%  \boris{we can do that, but do we want to? maybe put in a footnote?}
\end{theorem}
\begin{proof}
  We first construct a set $Q$ of $n$ points, no two of them with
  equal abscissa or ordinate, that admits a set of $2n-\Theta(\sqrt
  n)$ pairwise openly-disjoint rectangles, whose boundary contains two
  points from $Q$, which will imply the lower bound.

  Start with a grid of size $\sqrt n\times\sqrt n$, then rotate the
  whole grid infinitesimally clockwise, and finally perturb the points
  very slightly, so that no point coordinate is repeated and there are
  no collinearities.  The desired set of rectangles contains $B(p,q)$
  for every pair of points $p,q \in Q$ that were neighbors in the
  original grid; refer to Figure~\ref{fig:pairwiseDisjointRectangles}.

\begin{figure}[htbp!]
  \centering
  \includegraphics[scale=0.55]{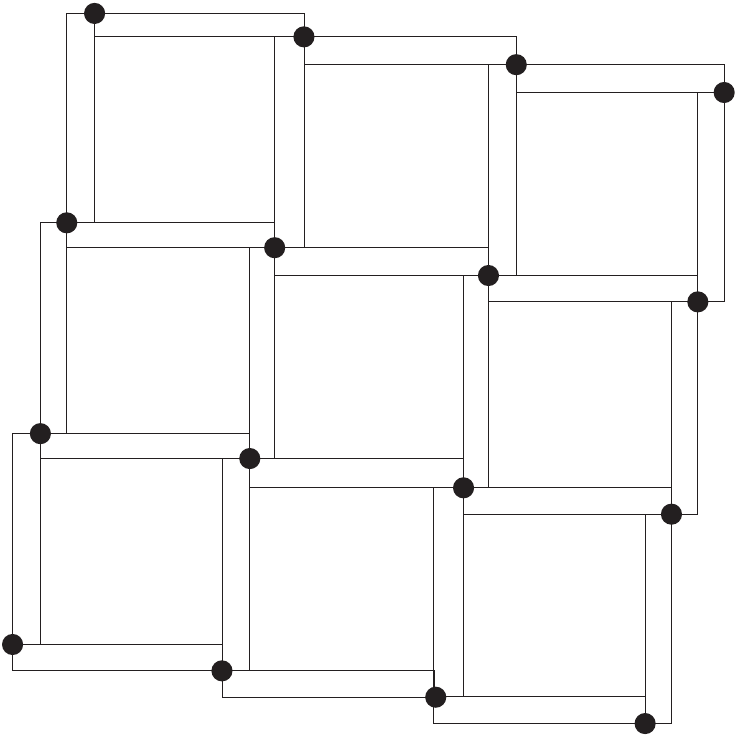}%
  %\hspace{0.5\linewidth}%
    \caption{Construction of a set $Q$ with a large number of rectangles
    with disjoint interiors, each one with two points from $Q$ on its
    boundary; every point of $Q$ participates in four rectangles, with
    the exception of those on the boundary of the configuration.}
  \label{fig:pairwiseDisjointRectangles}
\end{figure}

 \begin{figure}[htbp!]
  \centering
  \includegraphics[scale=0.9]{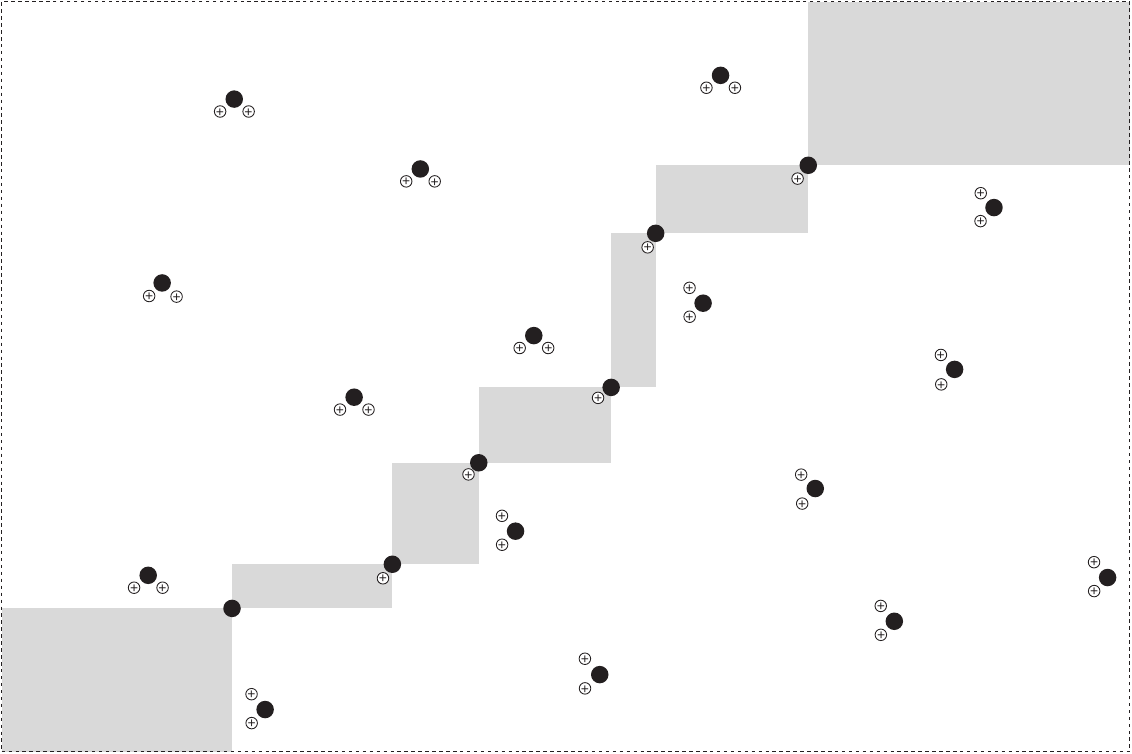}
  \caption{Construction, for a point set $P$, of a set $W$ of positive
    witnesses such that $\RG^+(P,W)=K_{|P|}$.}
  \label{fig:pairwiseDisjointRectangles2}
\end{figure}

For the proof that $2n-\Theta(\sqrt{n})$ points suffice to stab all the rectangles refer to Theorem 6 in \cite{ADH}. A sketch of the proof is
illustrated in Figure~\ref{fig:pairwiseDisjointRectangles2}. The points that get only one witness form an ascending chain.
A descending chain would be used similarly and one of these alternative chains must have
$\Omega(\sqrt{n})$ size.

 %(illustrated in Figure~\ref{fig:pairwiseDisjointRectangles2}).  
%\muriel{Should we put back this part of the proof? It just needs to be uncommented below in the tex file.}
\hfill $\Box$

  \end{proof}

%%%%%%%%%%%%%%%%%%%%%%%%%%%%%%%%%%%%%%%%%%%%%%%%%%%%%%%%%%%%%%%%%%%%%%%%%%%%%%

%%%%%%%%%%%%%%%%%%%%%%%%%%%%%%%%%%%%%%%%%%%%%%%%%%%%%%%%%%%%%%%%%%%%%%%%%%%%%%
%\section{Concluding remarks}\label{section:conclusion}
%Will be written some day.
%\ferran{... but I think we may skip this section in this occasion.}
%\boris{do we have something to say?}
%%%%%%%%%%%%%%%%%%%%%%%%%%%%%%%%%%%%%%%%%%%%%%%%%%%%%%%%%%%%%%%%%%%%%%%%%%%%%%
%\clearpage

%\small

\end{document}